\documentclass{article}
\usepackage{amsfonts}
\usepackage{amsmath}
\usepackage{fancyhdr}
\usepackage{fancybox}
\usepackage{color}
\usepackage{graphicx}
\newtheorem{theorem}{Theorem}

\newtheorem{proposition}[theorem]{Proposition}

\newenvironment{proof}[1][Proof]{\noindent\textbf{#1.} }{\ \rule{0.5em}{0.5em}}
\input{tcilatex}
\begin{document}

\title{Precise analytic treatment of Kerr and Kerr-(anti) de Sitter \\
black holes as gravitational lenses. }
\author{G. V. Kraniotis\thanks{%
Email: gkraniot@cc.uoi.gr.} \\
The University of Ioannina, Department of Physics, \\
Section of Theoretical Physics, GR-451 10 Ioannina, Greece.}
\maketitle

\begin{abstract}
The null geodesic equations that describe motion of photons in Kerr
spacetime are solved exactly in the presence of the cosmological constant $%
\Lambda .$ The exact solution for the deflection angle for generic light
orbits (i.e. non-polar, non-equatorial) is
calculated in terms of the generalized hypergeometric functions of Appell
and Lauricella.

We then consider the more involved issue in which the black hole acts as a
`gravitational lens'. The constructed Kerr black hole gravitational lens
geometry consists of an observer and a source located far away and placed at
arbitrary inclination with respect to black hole's equatorial plane. The
resulting lens equations are solved elegantly in terms of Appell-Lauricella
hypergeometric functions and the Weierstra\ss\ elliptic function. We then,
systematically, apply our closed form solutions for calculating the image
and source positions of generic photon orbits that solve the lens equations
and reach an observer located at various values of the polar angle for
various values of the Kerr parameter and the first integrals of motion. In
this framework, the magnification factors for generic orbits are calculated
in closed analytic form for the first time. The exercise is repeated with
the appropriate modifications for the case of non-zero cosmological constant.
\end{abstract}

\section{\protect\bigskip Introduction}

The issue of the bending of light (and the associated phenomenon of
gravitational lensing) from the gravitational field of a celectial body
(planet, star, black hole, galaxy) has been a very active and fruitful area
of research for fundamental physics. The Munich astronomer Johann Georg von
Soldner in 1801 \cite{Soldner} using Newtonian mechanics and assuming a
corpuscular theory for light derived \ a value for the deflection angle in
the Sun's gravitational field $\Delta \phi ^{S}=0^{\prime \prime}.875.$ Later, Einstein using the equations of general relativity \cite%
{Albert1}, derived a value of $\sim 1^{\prime \prime}.75$ which is consistent with the findings of Eddington's solar eclipse
experiment and subsequent measurements.

Despite the importance of the gravitational bending of light, in unravelling
the nature of the gravitational field and its cosmological implications not
many exact analytic results for the deflection angle of light orbits from
the gravitational field of important astrophysical objects are known in the
literature.

Recently, progress has been achieved \cite{Kraniotis} in obtaining the
closed form (strong-field) solution for the deflection angle of an
equatorial light ray in the Kerr gravitational field (spinning black hole,
rotating mass). Thus, going beyond the corresponding calculation for the
static gravitational field of a Schwarzschild black hole \cite{Albert2},\cite%
{ohanian}. More specifically, the closed form solution for the gravitational
bending of light for an equatorial photon orbit in Kerr spacetime was
derived and expressed elegantly in terms of Lauricella's hypergeometric
function $F_{D}.$ It was then applied, to calculate the deflection angle for
various values of the impact parameter and the spin of the galactic centre
black hole Sgr A*. The results exhibited clearly, the strong dependence of
the gravitational bending of light, on the spin of the black hole for small
values of impact parameter (frame dragging effects) \cite{Kraniotis}. In
addition, in \cite{Kraniotis}, the exact solution for (unstable) spherical
bound polar and non-polar photonic orbits was derived. However, the closed
form analytic solution for the important class of generic (i.e. non-polar
and non-equatorial) unbound light orbits were left out of the discussion in 
\cite{Kraniotis}.

One of the unsolved related important problems so far was the full analytic
treatment of the Kerr and Kerr-de Sitter black holes as gravitational
lenses. The closed form solution of this problem is imperative since the
Kerr black hole acts as a very strong gravitational lense and we may probe
general relativity, through the phenomenon of the bending of light induced
by the space time curvature of a spinning black hole, at the strong
gravitational field regime. A completely unexplored region of paramount
importance for fundamental physics and cosmology.

It is therefore the purpose of the present paper to calculate the exact
solution for the deflection angle for a generic photon orbit in the
asymptotically flat Kerr spacetime therefore generalizing the results in 
\cite{Kraniotis} and solve in closed analytic form the more involved problem
of treating the rotating black hole as a gravitational lense. The
constructed Kerr black hole gravitational lens geometry consists of an
observer and a source located far away and placed at arbitrary inclination
with respect to black hole's equatorial plane.

More specifically, we solve for the \textit{first time }in closed analytic
form, the resulting lens equations in Kerr geometry, in terms of the Weierstra%
\ss\ elliptic function $\wp (z)$, equation (\ref{WeierstrassKarl}), and in
terms of generalized hypergeometric functions of Appell-Lauricella equations
(\ref{Constraint1one}), (\ref{Lauricella4}), (\ref{KleistiGwniaki2}), (\ref%
{AzimuPigis}).

In addition, we calculate for the \textit{first time } exactly the resulting
magnification factors for generic light orbits in terms of the
hypergeometric functions of Appell and Lauricella. Our closed form form
solutions for the source and image potitions of the lens equations and the
corresponding magnification factors represent an important progress step in
the extraction of the phenomenological and astrophysical implications of
spinning black holes.

The resulting theory is of paramount importance for the galactic centre
studies given the strong experimental evidence we have from observation of
stellar orbits and flares, that the Sagittarius A$^{\ast }$ region, at the
galactic centre of Milky-Way, harbours a supermassive rotating black hole
with mass of 4 million solar masses \cite{Ghez}, \cite{GRAVITY}. Given the
fact that the GRAVITY experiment \cite{GRAVITY} and the proposed 30 metre
telescope (TMT) \cite{Ghez} aim at an accurary of 10 $\mu \mathrm{arcs}$ the
images calculated in this work and formed near the event horizon of the
spinning black hole should be a subject of experimental scrutiny.

Previous efforts on the issue of gravitational lensing  from a Kerr black
hole were concentrated on various approximations as well as numerical
techniques  using formal integrals \cite{SerenoBozza},\cite{VAzEste}.

The material of this work is organized as follows: in section \ref{NullGeo},
we present the null geodesics in a Kerr spacetime with a cosmological
constant. In section \ref{Kerrlensgeometry} we describe the Kerr-lens
geometry and relate the first integrals of motion to the observer's image
plane coordinates. In section \ref{megethynsy} we derive a formal expression
for the magnification using the Jacobian that relates observer's image plane
coordinates to the source position. This expression involves derivatives of
the lens equations in Kerr geometry and in our contribution we shall
calculate in closed form these derivatives in terms of the generalized
hypergeometric functions of Appell-Lauricella. In section \ref%
{Shadowonthewall}, we derive constraints from the condition that a photon
escapes to infinity and it is not caught in a (unstable) spherical orbit.
These constraints on the Carter's constant and impact factor define a region
usually called the shadow of the rotating black hole. For values of the
initial conditions inside the region enclosed by the boundary of the shadow
and the line with null value for Carter's constant there is no lensing
effect since the photons cannot escape and reach an observer. We also
discuss constraints arising from the polar motion. In section \ref{AktiOlok}%
, we derive for the first time the closed form solution for the angular
integrals involved in the gravitational Kerr lens, in terms of the
generalized hypergeometric functions of Appell-Lauricella. The full exact
solution for a light ray which originates from source's polar position and
involves $m-$polar inversions before reaching the polar coordinate of the
observer is derived: equations (\ref{GvniaKleisti1}),(\ref{KleistiGwniaki2}%
). In section \ref{AktiOlok}, we perform the analytic computation of the
radial integrals involved in the Kerr-lens in terms of the hypergeometric
functions of Appell and Lauricella. In the same section, we derive the
closed form solution for the source polar position in terms of the Weierstra%
\ss\ elliptic function $\wp (z,g_{2},g_{3})$ that implements the constraint
that arises from the first lens equation (\ref{AbelIntegral}). In sections %
\ref{IsimeParatiritis}, \ref{60moiresparatiritis}, we apply our exact
solutions for the calculation of the source and image positions for various
values of the spin of the black hole and the first integrals of motion, for
an equatorial observer and an observer located at a polar angle of $\pi /3$
respectively. We exhibit the image positions on the observer's image plane.
In Appendix \ref{Veta}, we collect the definition and the integral
representation of Lauricella's multivariable hypergeometric function $F_{D}.$
In addition in appendix \ref{Veta}, we prove in the form of Propositions,
some mathematical results concerning the transformation properties of the
function $F_{D}$ which are used in the main text.

\section{Null geodesics in a Kerr-(anti) de Sitter black hole.
\label{NullGeo}}

Taking into account the contribution from the cosmological constant 
$\Lambda $, the generalization of the Kerr solution is described by the
Kerr-de Sitter metric element which in Boyer-Lindquist (BL) coordinates is
given by \cite{Stuchlik}-\cite{Demianski}:

\begin{eqnarray}
\mathrm{d}s^{2} &=&\frac{\Delta _{r}}{\Xi ^{2}\rho ^{2}}(c\mathrm{d}t-a\sin
^{2}\theta \mathrm{d}\phi )^{2}-\frac{\rho ^{2}}{\Delta _{r}}\mathrm{d}r^{2}-%
\frac{\rho ^{2}}{\Delta _{\theta }}\mathrm{d}\theta ^{2}  \notag \\
&&-\frac{\Delta _{\theta }\sin ^{2}\theta }{\Xi ^{2}\rho ^{2}}(ac\mathrm{d}%
t-(r^{2}+a^{2})\mathrm{d}\phi )^{2}
\end{eqnarray}%
\begin{equation}
\Delta _{\theta }:=1+\frac{a^{2}\Lambda }{3}\cos ^{2}\theta ,\text{ }\Xi :=1+%
\frac{a^{2}\Lambda }{3}
\end{equation}

\begin{equation}
\Delta _{r}:=\left( 1-\frac{\Lambda }{3}r^{2}\right) \left(
r^{2}+a^{2}\right) -2\frac{GM}{c^{2}}r
\end{equation}

We denote by $a$ the rotation (Kerr) parameter and $M$ denotes the mass of
the spinning black hole.

The relevant null geodesic differential equations for the calculation of the
gravitational lensing effects (lens-equation) and for the calculation of the
deflection angle are:
 
\begin{equation}
\int^{r}\frac{\mathrm{d}r}{\pm \sqrt{R}}=\int^{\theta }\frac{\mathrm{d}%
\theta }{\pm \sqrt{\Theta }}  \label{AbelIntegral}
\end{equation}%
\begin{equation}
\Delta \phi =\int \mathrm{d}\phi =\int^{\theta }-\frac{\Xi ^{2}}{\pm \Delta
_{\theta }\sin ^{2}\theta }\frac{(a\sin ^{2}\theta -\Phi )\mathrm{d}\theta }{%
\sqrt[2]{\Theta }}+\int^{r}\frac{a\Xi ^{2}}{\pm \Delta _{r}}%
[(r^{2}+a^{2})-a\Phi ]\frac{\mathrm{d}r}{\sqrt[2]{R}}
\label{Deflectionangleazi}
\end{equation}%
where

\begin{equation}
R:=\left\{ \Xi ^{2}\left[ (r^{2}+a^{2})-a\Phi \right] ^{2}-\Delta _{r}\left[
\Xi ^{2}\left( \Phi -a\right) ^{2}+\mathcal{Q}\right] \right\}
\label{quartic1}
\end{equation}%
and

\begin{equation}
\Theta :=\left\{ [\mathcal{Q}+(\Phi -a)^{2}\Xi ^{2}]\Delta _{\theta }-\frac{%
\Xi ^{2}(a\sin ^{2}\theta -\Phi )^{2}}{\sin ^{2}\theta }\right\}
\label{gwnia}
\end{equation}

We also derive the equation related to time-delay:%
\begin{equation}
ct=\int^{r}\frac{\Xi ^{2}(r^{2}+a^{2})\left[ (r^{2}+a^{2})-\Phi a\right] }{%
\pm \Delta _{r}\sqrt{R}}\mathrm{d}r-\int^{\theta }\frac{a\Xi ^{2}(a\sin
^{2}\theta -\Phi )}{\pm \Delta _{\theta }\sqrt{\Theta }}\mathrm{d}\theta
\end{equation}

The parameters $\Phi ,\mathcal{Q}$ are associated to the first integrals of
motion. The former is the impact parameter and the latter is related to the
hidden  first integral (due to the separation of variables in the
corresponding Hamilton-Jacobi partial differential equation (PDE)).

\section{ The Kerr black hole as a gravitational lens.
\label{Kerrlensgeometry}}

\subsection{Observer's image plane}

Assume without loss of generality that the observer's position is at $%
(r_{O},\theta _{O},0).$ Likewise, for the source we have $(r_{S},\theta
_{S},\phi _{S})$ . We also assume in this section that $\Lambda =0.$ In the
observer's reference frame, an incoming light ray is described by a
parametric curve $x(r),y(r),z(r),$ where $r^{2}=x^{2}+y^{2}+z^{2}.$ For
large $r$ this the usual radial BL coordinate. At the location of the
observer, the tangent vector to the parametric curve is given by: $(\mathrm{d%
}x/\mathrm{d}r)|_{r_{O}}$ $\widehat{\mathbf{x}}+(\mathrm{d}y/\mathrm{d}%
r)|_{r_{O}}\widehat{\mathbf{y}}+(\mathrm{d}z/\mathrm{d}r)|_{r_{O}}\widehat{%
\mathbf{z}}.$ This \ vector describes a straight line which intersects the $%
(\alpha ,\beta )$ plane or \textit{observer's image plane }as it is usually
called \cite{BARDEEN}-\cite{VAzEste} at $(\alpha _{i},\beta _{i})$ see fig.\ref{GravRotLens}$.$

\begin{figure}
\begin{center}
\includegraphics{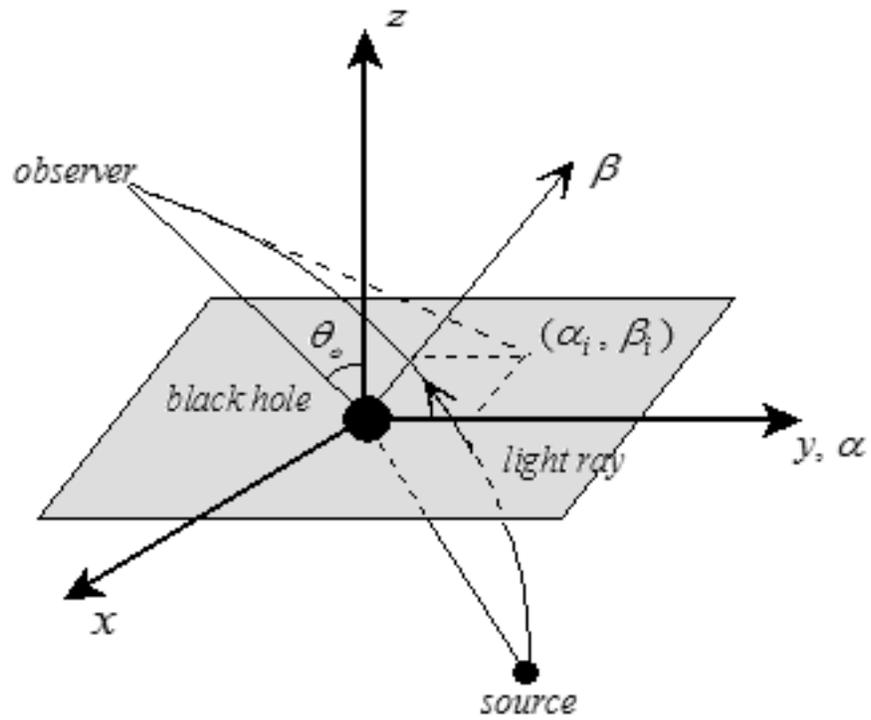}
\caption{The Kerr black hole gravitational lens geometry. The reference 
frame is chosen so that, as seen from infinity, the black hole is rotating around the $z$-axis. \label{GravRotLens}} 
\end{center}
\end{figure}

The point $(\alpha _{i},\beta _{i})$ is the point $(-\beta _{i}\cos \theta
_{O},\alpha _{i},\beta _{i}\cos \theta _{O})$ in the \ $(x,y,z)$\ system.
Our purpose now is to relate the $\alpha _{i},\beta _{i}$ variables to the
first integrals of motion $\Phi ,\mathcal{Q}.$ For this we need to use the
equation of straight line in space. A straight line can be defined from a
point $P_{1}(x_{1},y_{1},z_{1})$ on it and a vector $\overline{\epsilon }%
(\epsilon _{1,}\epsilon _{2,}\epsilon _{3})$ parallel to it. The analytic
equations of straight line are then:

\begin{equation}
\frac{x-x_{1}}{\epsilon _{1}}=\frac{y-y_{1}}{\epsilon _{2}}=\frac{z-z_{1}}{%
\epsilon _{3}}  \label{Eythia}
\end{equation}

Applying (\ref{Eythia}) we derive the equations:

\begin{equation}
\frac{-\beta _{i}\cos \theta _{O}-r_{O}\sin \theta _{O}}{r_{O}\cos \theta
_{O}\frac{\mathrm{d}\theta }{\mathrm{d}r}|_{r=r_{0}}+\sin \theta _{O}}=\frac{%
\alpha _{i}}{r_{O}\sin \theta _{O}\frac{\mathrm{d}\phi }{\mathrm{d}r}%
|_{r=r_{O}}}=\frac{\beta _{i}\cos \theta _{O}-r_{O}\cos \theta _{O}}{\cos
\theta _{O}-r_{O}\sin \theta _{O}\frac{\mathrm{d}\theta }{\mathrm{d}r}%
|_{r=r_{O}}}
\end{equation}

Solving for $\alpha _{i},\beta _{i}$ we obtain the equations:%
\begin{eqnarray}
\alpha _{i} &=&-r_{O}^{2}\sin \theta _{O}\frac{\mathrm{d}\phi }{\mathrm{d}r}%
|_{r=r_{O}} \\
\beta _{i} &=&r_{O}^{2}\frac{\mathrm{d}\theta }{\mathrm{d}r}|_{r=r_{O}}
\end{eqnarray}

Now we have from the null geodesics that:%
\begin{equation}
\frac{\mathrm{d}\theta }{\mathrm{d}r}|_{r=r_{O}}=\frac{\Theta (\theta
_{O})^{1/2}}{R(r_{O})^{1/2}}  \label{geodidat}
\end{equation}

and%
\begin{equation}
\frac{\mathrm{d}\phi }{\mathrm{d}r}|_{r=r_{O}}=\frac{\Phi }{\sqrt[2]{R(r_{O})%
}}\frac{1}{\sin ^{2}(\theta _{O})}+\frac{2aGM\frac{r_{O}}{c^{2}}-a^{2}\Phi }{%
r_{O}^{2}\left[ 1+\frac{a^{2}}{r_{O}^{2}}-\frac{2GM}{r_{O}c^{2}}\right] }%
\frac{1}{\sqrt[2]{R(r_{O})}}  \label{geode}
\end{equation}

Using eqns(\ref{geodidat}),(\ref{geode}) and assuming large observer's
distance $r_{O}($i.e. $r_{O}\longrightarrow \infty )$ we derive simplified
expressions relating the coordinates ($\alpha _{i},\beta _{i})$ on the
observer's image plane to the integrals of motion%
\begin{eqnarray}
\Phi &\simeq &-\alpha _{i}\sin \theta _{O}  \label{ObserSkym} \\
\mathcal{Q} &\simeq &\beta _{i}^{2}+(\alpha _{i}^{2}-a^{2})\cos ^{2}(\theta
_{O})  \label{Integr}
\end{eqnarray}

We can also express the position of the source on the observer's sky in
terms of its coordinates $(r_{S},\theta _{S},\phi _{S})$ and the observer
coordinates. Indeed, the equation for a straight line can be determined by
two points $P_{1}(x_{1},y_{1},z_{1}),P_{2}(x_{2},y_{2},z_{2})$:

\begin{equation}
\frac{x-x_{1}}{x_{2}-x_{1}}=\frac{y-y_{1}}{y_{2}-y_{1}}=\frac{z-z_{1}}{%
z_{2}-z_{1}}
\end{equation}%
Thus applying the above formula for the straight line connecting the
observer and the source yields the equations:

\begin{eqnarray}
\alpha _{S} &=&\frac{r_{O}r_{S}\sin \theta _{S}\sin \phi _{S}}{%
r_{O}-r_{S}(\cos \theta _{S}\cos \theta _{O}+\sin \theta _{O}\sin \theta
_{S}\cos \phi _{S})}  \notag \\
\beta _{S} &=&\frac{-r_{O}r_{S}(\sin \theta _{O}\cos \theta _{S}-\sin \theta
_{S}\cos \phi _{S}\cos \theta _{O})}{r_{O}-r_{S}(\cos \theta _{S}\cos \theta
_{O}+\sin \theta _{O}\sin \theta _{S}\cos \phi _{S})}  \label{sourceskycoord}
\end{eqnarray}

\section{\protect\bigskip Magnification factors and positions of images\label%
{megethynsy}.}

In the following sections, we shall perform a detailed novel calculation of
the lens effect for the deflection of light produced by the gravitational
field of a rotating (Kerr) black hole and a cosmological Kerr black hole
(i.e. for non-zero cosmological constant $\Lambda $).

The flux of an image of an infinitesimal source \ is the product of its
surface brightness and the solid angle $\Delta \omega $ it subtends on the
sky. \ Since the former quantity is unchanged during light deflection, the
ratio of the flux of a sufficiently small image to that of its corresponding
source in the absence of the lens, is given by

\begin{equation}
\mu =\frac{\Delta \omega }{\left( \Delta \omega \right) _{0}}=\frac{1}{|J|}
\end{equation}%
where 0-subscripts denote undeflected quantities \cite{Scneider} and $J$ is
the Jacobian of the transformation $(x_{S,}\ y_{S})\rightarrow (x_{i},y_{i})$
\footnote{%
Recall in the small angles approxiamation: $\alpha _{i}\approx r_{O}$ $%
x_{i}, $ $\beta _{i}\approx r_{O}$ $y_{i}.$ Also we define: $x_{S}:=\frac{%
\alpha _{S}}{r_{O}},y_{S}:=\frac{\beta _{S}}{r_{O}}.$}. Writting $%
x_{S}=x_{S}(x_{i},y_{i}),y_{S}=y_{S}(x_{i},y_{i})$ we can find expressions
for the partial derivatives appearing in the Jacobian by differentiating
equations (\ref{AbelIntegral}) and (\ref{Deflectionangleazi}). Indeed, the
Jacobian is given by the expression:%
\begin{equation}
J=xw-zy
\end{equation}%
where we defined: $x:=\frac{\partial x_{S}}{\partial x_{i}},y:=\frac{%
\partial x_{S}}{\partial y_{i}},z:=\frac{\partial y_{S}}{\partial x_{i}},w:=%
\frac{\partial y_{S}}{\partial y_{i}}.$ Writting equations (\ref%
{AbelIntegral}) and (\ref{Deflectionangleazi}) as follows:%
\begin{eqnarray}
R_{1}(x_{i},y_{i})-A_{1}(x_{i},y_{i},x_{S},y_{S},m) &=&0  \notag \\
\Delta \phi
(x_{S},y_{S},n)-R_{2}(x_{i},y_{i})-A_{2}(x_{i},y_{i},x_{S},y_{S},m) &=&0
\label{Lenseksiswsi}
\end{eqnarray}%
we set up the following system of equations:%
\begin{eqnarray}
\beta _{1} &=&-\alpha _{1}x-\alpha _{2}z \\
\beta _{2} &=&-\alpha _{1}y-\alpha _{2}w \\
-\beta _{3} &=&\alpha _{3}x+\alpha _{4}z \\
-\beta _{4} &=&\alpha _{3}y+\alpha _{4}w
\end{eqnarray}%
where 
\fbox{$\displaystyle \alpha_1=\frac{\partial A_1}{\partial x_S},\; 
\alpha_2=\frac{\partial A_1}{\partial y_S}, \;
\alpha_3=-\frac{\partial \phi_s}{\partial x_S}-\frac{\partial A_2}{\partial x_S},\;
\alpha_4=-\frac{\partial \phi_s}{\partial y_S}-\frac{\partial A_2}{\partial y_S}$}%
,

$%
\fbox{$\displaystyle 
\beta_1=\frac{\partial R_1}{\partial x_i}-\frac{\partial A_1}{\partial x_i},\;
\beta_2=\frac{\partial R_1}{\partial y_i}-\frac{\partial A_1}{\partial y_i},\;
\beta_3=\frac{\partial R_2}{\partial x_i}+\frac{\partial A_2}{\partial x_i},\;
\beta_4=\frac{\partial R_2}{\partial y_i}+\frac{\partial A_2}{\partial y_i}$}%
.$

Solving for $x,y,z,w$ we obtain:

\begin{equation}
\mu =\frac{1}{|J|}=\left\vert \frac{\alpha _{1}\alpha _{4}-\alpha _{2}\alpha
_{3}}{\beta _{1}\beta _{4}-\beta _{2}\beta _{3}}\right\vert
\label{MegenthysiVaritiki}
\end{equation}

The parameters $n=0,1,2,\ldots $ and $m=0,1,2,\ldots $ are the number of
windings around the $z$ axis and the number of turning points in the polar
coordinate $\theta $ respectively. We shall discuss the latter in detail in
the section that follows.

\section{ The boundary of the shadow of the rotating black
hole and constraints on the parameter space\label{Shadowonthewall}.}

The condition for a photon to escape to infinity , which is also the
condition for the spherical photon orbits in Kerr spacetime \cite{Kraniotis}%
, is given by the vanishing of the quartic polynomial \ $R(r)$ and its first
derivative (also in this case $\frac{d^{2}R}{dr^{2}}|_{r=r_{f}}>0$).
Implementing these two conditions, expressions for the parameter $\Phi $ and
Carter's constant $\mathcal{Q}$ are obtained \cite{Kraniotis}, \cite{Teo}:%
\begin{equation}
\Phi =\frac{a^{2}\frac{GM}{c^{2}}+a^{2}r-3\frac{GM}{c^{2}}r^{2}+r^{3}}{%
a\left( \frac{GM}{c^{2}}-r\right) },\qquad \mathcal{Q}=-\frac{r^{3}\left(
-4a^{2}\frac{GM}{c^{2}}+r\left( \frac{-3GM}{c^{2}}+r\right) ^{2}\right) }{%
a^{2}\left( \frac{GM}{c^{2}}-r\right) ^{2}}  \label{Photonsphere}
\end{equation}

The perturbed, from the radius $r=r_{inst},$ of unstable spherical null
orbits in Kerr spacetime, and thus escaped photon, will be detected on the
observer's image plane, at the coordinates :%
\begin{eqnarray}
x_{i} &=&\frac{a^{2}(r+\frac{GM}{c^{2}})+r^{2}(r-\frac{3GM}{c^{2}})}{%
r_{O}\sin \theta _{O}a\left( r-\frac{GM}{c^{2}}\right) },  \notag \\
y_{i} &=&\frac{\pm \sqrt{-r^{3}[r(r-\frac{3GM}{c^{2}})^{2}-4a^{2}\frac{GM}{%
c^{2}}]-2a^{2}r(2a^{2}\frac{GM}{c^{2}}+r^{3}-3r\frac{G^{2}M^{2}}{c^{4}}%
)z_{O}-a^{4}(r-\frac{GM}{c^{2}})^{2}z_{O}^{2}}}{r_{O}\sin \theta _{O}a\left(
r-\frac{GM}{c^{2}}\right) }  \notag \\
&&  \label{ShadowBH}
\end{eqnarray}%
Equations (\ref{ShadowBH}) were derived by plugging into equations (\ref%
{ObserSkym}),(\ref{Integr}) the values of the parameters $\mathcal{Q},\Phi $ 
$\ $that correspond to the conditions for the photon to escape to infinity,
equations (\ref{Photonsphere}). A photon will be detected when the argument
of the square root in eqn.(\ref{ShadowBH}) is positive. In Eqn(\ref{ShadowBH}%
), $z_{O}:=\cos ^{2}\theta _{O}.$

\bigskip

With the aid of equations (\ref{ObserSkym}) and (\ref{Integr}) we derive:

\begin{equation}
\alpha _{i}^{2}+\beta _{i}^{2}=\Phi ^{2}+\mathcal{Q}+a^{2}z_{O}
\label{Radiusgeq4}
\end{equation}

Apart from the constraints expressed by equations (\ref{Photonsphere}) we
also derive \ constraints for the motion of light from the allowed polar
region:$\theta _{\min }\leq \theta _{S},\theta _{O}\leq \theta _{\max }.$
Indeed using the variable $z_{j}:=\cos ^{2}\theta _{j}$, we have $z_{m}\geq
z_{O}$ where $z_{m}$ is the positive root of:%
\begin{equation*}
-a^{2}z_{m}^{2}+(a^{2}-\mathcal{Q}-\Phi ^{2})z_{m}+\mathcal{Q}=0
\end{equation*}

Let us see how this can be understood. Defining: $z_{m}:=z_{O}-x$ we derive
the quadratic equation for $x$%
\begin{equation}
-a^{2}x^{2}-x(a^{2}-\mathcal{Q}-\Phi
^{2}-2a^{2}z_{O})-a^{2}z_{O}^{2}+z_{O}(a^{2}-\mathcal{Q}-\Phi ^{2})+\mathcal{%
Q}=0
\end{equation}%
with roots: 
\begin{eqnarray}
x_{1,2} &=&\frac{-a^{2}+\mathcal{Q}+2a^{2}z_{O}+\Phi ^{2}\mp \sqrt[2]{4a^{2}%
\mathcal{Q}+(-a^{2}+\mathcal{Q}+\Phi ^{2})^{2}}}{2a^{2}}  \notag \\
&=&\frac{(\alpha _{i}^{2}+\beta _{i}^{2})-a^{2}w_{O}\mp \sqrt[2]{\left(
(\alpha _{i}^{2}+\beta _{i}^{2})-a^{2}w_{O}\right) ^{2}+4a^{2}\beta
_{i}^{2}w_{O}}}{2a^{2}}
\end{eqnarray}%
where $w_{O}:=\sin ^{2}\theta _{O}.$ The "radius" $\Phi ^{2}+\mathcal{Q}$
must be greater or equal than the boundary of the photon region defined by
Eqs.(\ref{Photonsphere}) and the line $\mathcal{Q}=0.$ The minimum of this
value is reached when $\mathcal{Q}=0$ and $a\rightarrow 1.$The actual
minimum value is ($\Phi ^{2}(r)+\mathcal{Q(}r))_{\min }=4.$ Thus, by Eq.(\ref%
{Radiusgeq4}) we have that $\alpha _{i}^{2}+\beta _{i}^{2}\geq 4,$ and since 
$0\leq a^{2}w_{O}\leq 1,$ it follows the inequality $a^{2}w_{O}-(\alpha
_{i}^{2}+\beta _{i}^{2})<0$ and consequently, $x\leq 0.$ Thus we conclude
that $z_{m}\geq z_{O}.$ Similar arguments ensure that when $z_{S}>z_{O}$ it
follows $z_{m}\geq z_{S}$ $\cite{VAzEste}.$

\bigskip

\section{Closed form solution for the angular integrals\label{olokpolar}.}

Let us perform now the exact computation of the angular integrals which
occur in the generic photon orbits in Kerr spacetime thereby generalizing
the results of \cite{Kraniotis}. In the case under investigation, we have to
take into account the 
{\bf turning\;points}
in the polar coordinate. A generic angular polar integral can be written:%
\begin{equation}
\pm \int_{\theta _{1}}^{\theta _{2}}=\int_{\mathrm{\min (}z_{1},z_{2})}^{%
\mathrm{\max (}z_{1},z_{2})}+[1-\mathrm{sign(}\theta _{1}\circ \theta
_{2})]\int_{0}^{\mathrm{\min (}z_{1},z_{2})}
\end{equation}%
where:%
\begin{equation}
\theta _{1}\circ \theta _{2}:=\cos \theta _{1}\cos \theta _{2}
\end{equation}%
Indeed, using the variable 
\fbox{$\displaystyle z:=\cos^2\theta$}
we derive:%
\begin{equation}
-\frac{1}{2}\frac{\mathrm{d}z}{\sqrt{z}}\frac{1}{\sqrt{1-z}}=\mathrm{sign(}%
\frac{\pi }{2}-\theta )\mathrm{d}\theta
\end{equation}%
This is the result of the fact that in the interval $0\leq \theta \leq \frac{%
\pi }{2},$ $\cos \theta \geq 0$ and $\sin \theta \geq 0,$ while in the
interval $\frac{\pi }{2}\leq \theta \leq \pi ,$ $\sin \theta \geq 0$, $\cos
\theta \leq 0.$ The angular integration in the polar variable includes the
terms:%
\begin{equation}
\int^{\theta }=\pm \int_{\theta _{S}}^{\theta _{\min /\max }}\pm
\int_{\theta _{\min /\max }}^{\theta _{\max /\min }}\pm \int_{\theta
_{_{\max /\min }}}^{\theta _{_{\min /\max }}}\pm \cdots \pm \int_{\theta
_{\max /\min }}^{\theta _{O}}  \label{GonTot}
\end{equation}

The roots $z_{m},z_{3}$ (of 
\fbox{$\displaystyle \Theta(\theta)=0$}%
) are expressed in terms of the integrals of motion and the cosmological
constant by the expressions:

\bigskip

\begin{equation}
z_{m,3}=\frac{\mathcal{Q}+\Phi ^{2}\Xi ^{2}-H^{2}\pm \sqrt{(\mathcal{Q}+\Phi
^{2}\Xi ^{2}-H^{2})^{2}+4H^{2}\mathcal{Q}}}{-2H^{2}}  \label{cubicroots}
\end{equation}

and

\begin{equation}
H^{2}:=\frac{a^{2}\Lambda }{3}[\mathcal{Q}+(\Phi -a)^{2}\Xi ^{2}]+a^{2}\Xi
^{2}
\end{equation}%
For $\Lambda =0$, the turning points take the form:%
\begin{equation}
\fbox{$\displaystyle z_m=\frac{a^2-{\cal Q}-\Phi^2+
\sqrt{4a^2 {\cal Q}+(-a^2+{\cal Q}+\Phi^2)^2}}{2a^2}, $}
\end{equation}
where the subscript \textquotedblleft m\textquotedblright\ \ stands for
\textquotedblleft \textrm{min/max\textquotedblright . }The corresponding
angles are:%
\begin{equation}
\fbox{$\displaystyle \theta_{\rm min/max}={\rm Arccos}(\pm\sqrt{z_m}) $}
\end{equation}%

Now for $\theta _{j}$ and $\theta _{\mathrm{\min /\max }}$ in the same
hemisphere:%
\begin{equation}
\int_{\theta _{j}}^{\theta _{\mathrm{\min /\max }}}\frac{\mathrm{d}\theta }{%
\pm \sqrt[2]{\Theta (\theta )}}=\frac{1}{2|a|}\int_{z_{j}}^{z_{m}}\frac{%
\mathrm{d}z}{\sqrt[2]{z(z_{m}-z)(z-z_{3})}}\equiv I_{3}  \label{Firstang}
\end{equation}%
Let us now calculate the elliptic integral in eqn.(\ref{Firstang}) in 
{\em closed\;analytic\;form}%
. Applying the transformation:

\begin{equation}
z=z_{m}+\xi ^{2}(z_{j}-z_{m})  \label{PrGwnmetas}
\end{equation}%
our integral is calculated in closed form in terms of Appell's generalized
hypergeometric function $F_{1}$ of two variables:

\begin{equation}
I_{3}=\frac{1}{2|a|}\frac{\sqrt[2]{(z_{m}-z_{j})}}{\sqrt[2]{%
z_{m}(z_{m}-z_{3})}}F_{1}\left( \frac{1}{2},\frac{1}{2},\frac{1}{2},\frac{3}{%
2},\frac{z_{m}-z_{j}}{z_{m}},\frac{z_{m}-z_{j}}{z_{m}-z_{3}}\right) \frac{%
\Gamma (\frac{1}{2})\Gamma (1)}{\Gamma (3/2)}
\end{equation}%
On the other hand using the transformation:%
\begin{equation}
\fbox{$\displaystyle z=\frac{uz_jz_m-z_jz_m}{uz_j-z_m} $}
\end{equation}
we calculate in closed form:
\begin{eqnarray}
&&\frac{1}{2\left\vert a\right\vert }\int_{0}^{z_{j}}\frac{\mathrm{d}
z}{\sqrt[2]{z(z_{m}-z)(z-z_{3})}}  \notag \\
&=&\frac{1}{\left\vert a\right\vert }\frac{\sqrt[2]{z_{j}}}{z_{m}}\sqrt[2]{
\frac{z_{j}-z_{m}}{z_{3}-z_{j}}}F_{1}\left( 1,\frac{1}{2},\frac{1}{2},\frac{3
}{2},\frac{z_{j}}{z_{m}},\frac{z_{j}(z_{m}-z_{3})}{z_{m}(z_{j}-z_{3})}\right)
\notag \\
&=&\frac{1}{\left\vert a\right\vert }\frac{\sqrt[2]{\frac{z_{j}(z_{m}-z_{3})}{z_{m}(z_{j}-z_{3})}}}{\sqrt[2]{z_{m}-z_{3}}}F_{1}\left( \frac{1}{2},\frac{1}{2},\frac{1}{2},\frac{3}{2},\frac{z_{m}}{z_{m}-z_{3}}\frac{z_{j}(z_{m}-z_{3})}{z_{m}(z_{j}-z_{3})},\frac{z_{j}(z_{m}-z_{3})}{z_{m}(z_{j}-z_{3})}\right)  \notag  \label{angulAppell} \\
&&  \label{XrisiProp1}
\end{eqnarray}
In going from the second line to the third of (\ref{XrisiProp1}) we made use
of the following identity  of Appell's first generalised hypergeometric
function of two variables:

\begin{equation}
F_{1}(\alpha ,\beta ,\beta ^{\prime },\gamma ,x,y)=(1-x)^{-\beta
}(1-y)^{\gamma -\alpha -\beta ^{\prime }}F_{1}(\gamma -\alpha ,\beta ,\gamma
-\beta -\beta ^{\prime },\gamma ,\frac{x-y}{x-1},y)
\end{equation}

Likewise we derive the closed form solution for the following integral:

\begin{eqnarray}
&&\frac{1}{2|a|}\int_{0}^{z_{j}}\frac{\mathrm{d}z}{(1-z)\sqrt[2]{%
z(z_{m}-z)(z-z_{3})}}  \notag \\
&=&\frac{z_{j}}{z_{m}}\frac{1}{|a|}\frac{z_{j}-z_{m}}{1-z_{j}}\frac{1}{\sqrt[%
2]{z_{j}(z_{j}-z_{m})(z_{3}-z_{j})}}\times  \notag \\
&&F_{D}\left( 1,1,-\frac{1}{2},\frac{1}{2},\frac{3}{2},\frac{z_{j}(1-z_{m})}{%
z_{m}(1-z_{j})},\frac{z_{j}}{z_{m}},\frac{z_{j}(z_{m}-z_{3})}{%
z_{m}(z_{j}-z_{3})}\right)  \notag \\
&=&\frac{1}{|a|}\frac{z_{j}}{z_{m}}\sqrt[2]{\frac{z_{m}}{-z_{3}z_{j}}}%
F_{D}\left( \frac{1}{2},1,\frac{1}{2},\frac{1}{2},\frac{3}{2},z_{j},\frac{%
z_{j}}{z_{m}},\frac{z_{j}}{z_{3}}\right)  \label{DeuxAnguvier}
\end{eqnarray}%
Producing the last line of equation (\ref{DeuxAnguvier}) we used the
following formula for the 
{\color{blue}{Lauricella\;function}}
{\color{red}{$\displaystyle F_D$}}%
:

\begin{proposition}
\begin{eqnarray}
F_{D}(\alpha ,\beta ,\beta ^{\prime },\beta ^{\prime \prime },\gamma ,x,y,z)
&=&(1-y)^{\gamma -\alpha -\beta ^{\prime }}(1-x)^{-\beta }(1-z)^{-\beta
^{\prime \prime }}\times  \notag \\
&&F_{D}\left( \gamma -\alpha ,\beta ,\gamma -\beta -\beta ^{\prime }-\beta
^{\prime \prime },\beta ^{\prime \prime },\gamma ,\frac{x-y}{x-1},y,\frac{z-y%
}{z-1}\right)  \notag
\end{eqnarray}
\end{proposition}

\begin{proof}
Applying the transformation:%
\begin{equation}
u=\frac{1-\nu }{1-\nu y}
\end{equation}%
onto the integral:%
\begin{equation}
IR_{F_{D}}=\int_{0}^{1}u^{\alpha -1}(1-u)^{\gamma -\alpha -1}(1-ux)^{-\beta
}(1-u\text{ }y)^{-\beta ^{\prime }}(1-uz)^{-\beta ^{\prime \prime }}\mathrm{d%
}u  \label{IRT1}
\end{equation}%
we derive:%
\begin{eqnarray}
(1-u)^{\gamma -\alpha -1} &=&\left( \frac{\nu (1-y)}{1-\nu y}\right)
^{\gamma -\alpha -1},\text{\quad }(1-ux)^{-\beta }=\left( \frac{(1-x)[1-%
\frac{\nu (x-y)}{(x-1)}]}{1-\nu y}\right) ^{-\beta }  \notag \\
(1-u\text{ }y)^{-\beta ^{\prime }} &=&\frac{(1-y)^{-\beta ^{\prime }}}{%
(1-\nu y)^{-\beta ^{\prime }}},\quad (1-uz)^{-\beta ^{\prime \prime
}}=\left( \frac{(1-z)[1-\frac{\nu (z-y)}{z-1}]}{1-\nu y}\right) ^{-\beta
^{\prime \prime }}
\end{eqnarray}%
and thus we obtain the result: 
\begin{eqnarray}
IR_{F_{D}} &=&(1-y)^{\gamma -\alpha }(1-x)^{-\beta }(1-y)^{-\beta ^{\prime
}}(1-z)^{-\beta ^{\prime \prime }}\times  \notag \\
&&\int_{0}^{1}\mathrm{d}\nu \text{ }\nu ^{\gamma -\alpha -1}(1-\nu )^{\alpha
-1}(1-\nu \text{ }y)^{-(\gamma -\beta -\beta ^{\prime }-\beta ^{\prime
\prime })}(1-\nu \frac{x-y}{x-1})^{-\beta }(1-\nu \frac{z-y}{z-1})^{-\beta
^{\prime \prime }}  \notag \\
&&
\end{eqnarray}%
or 
\begin{eqnarray}
F_{D}(\alpha ,\beta ,\beta ^{\prime },\beta ^{\prime \prime },\gamma ,x,y,z)
&=&(1-y)^{\gamma -\alpha -\beta ^{\prime }}(1-x)^{-\beta }(1-z)^{-\beta
^{\prime \prime }}\times  \notag \\
&&F_{D}\left( \gamma -\alpha ,\beta ,\gamma -\beta -\beta ^{\prime }-\beta
^{\prime \prime },\beta ^{\prime \prime },\gamma ,\frac{x-y}{x-1},y,\frac{z-y%
}{z-1}\right)  \notag
\end{eqnarray}
\end{proof}

Likewise applying (\ref{PrGwnmetas}):%
\begin{eqnarray}
I_{4} &:&=\frac{\Phi }{2|a|}\int_{z_{j}}^{z_{m}}\frac{\mathrm{d}z}{%
(1-z)\sqrt[2]{z(z_{m}-z)(z-z_{3})}}  \notag \\
&=&\frac{\Phi }{2|a|}\sqrt[2]{\frac{(z_{m}-z_{j})}{z_{m}}}\frac{1}{\sqrt[2]{%
(z_{m}-z_{3})}}\frac{2}{(1-z_{m})}F_{D}\left( \frac{1}{2},1,\frac{1}{2},%
\frac{1}{2},\frac{3}{2},\frac{z_{j}-z_{m}}{1-z_{m}},\frac{z_{m}-z_{j}}{z_{m}}%
,\frac{z_{m}-z_{j}}{z_{m}-z_{3}}\right)  \notag \\
&&  \label{Angular4}
\end{eqnarray}

Let us now compute exactly the term: $\pm \int_{\theta _{\mathrm{\min /\max }%
}}^{\theta _{\mathrm{\max /\min }}},$ in (\ref{GonTot}):%
\begin{equation}
\pm \int_{\theta _{\mathrm{\min /\max }}}^{\theta _{\mathrm{\max /\min }%
}}=2\int_{0}^{z_{m}}
\end{equation}%
since $\cos ^{2}\theta _{\mathrm{\min /\max }}=z_{m}$ and $\theta _{\mathrm{%
\min }}\circ \theta _{\mathrm{\max }}=-z_{m}.$

\bigskip Equation (\ref{Angular4}) for $z_{j}=0,$ becomes :%
\begin{eqnarray}
&&\frac{\Phi }{2|a|}\frac{1}{\sqrt[2]{(z_{m}-z_{3})}}\frac{2}{(1-z_{m})}%
F_{D}\left( \frac{1}{2},1,\frac{1}{2},\frac{1}{2},\frac{3}{2},\frac{-z_{m}}{%
1-z_{m}},1,\frac{z_{m}}{z_{m}-z_{3}}\right)  \notag \\
&=&\frac{\Phi }{|a|}\frac{1}{\sqrt[2]{(z_{m}-z_{3})}}\frac{1}{(1-z_{m})}%
\frac{\pi }{2}F_{1}\left( \frac{1}{2},1,\frac{1}{2},1,\frac{-z_{m}}{1-z_{m}},%
\frac{z_{m}}{z_{m}-z_{3}}\right)  \notag \\
&=&\frac{\Phi }{|a|}\frac{1}{\sqrt[2]{(z_{m}-z_{3})}}\frac{\pi }{2}%
F_{1}\left( \frac{1}{2},1,-\frac{1}{2},1,\frac{z_{m}(1-z_{3})}{z_{m}-z_{3}},%
\frac{z_{m}}{z_{m}-z_{3}}\right)  \notag \\
&=&\frac{\Phi }{|a|}\frac{1}{\sqrt[2]{(z_{m}-z_{3})}}\frac{\pi }{2}\frac{1}{%
1-z_{3}}\left( F(\frac{1}{2},\frac{1}{2},1,\frac{z_{m}}{z_{m}-z_{3}}%
)-z_{3}F_{1}\left( \frac{1}{2},1,\frac{1}{2},1,\frac{z_{m}(1-z_{3})}{%
z_{m}-z_{3}},\frac{z_{m}}{z_{m}-z_{3}}\right) \right)  \notag \\
&&  \label{SovaroTP}
\end{eqnarray}%
On the other hand the angular integrals of the form $\pm \int_{\theta
_{S}}^{\theta _{\mathrm{\min /\max }}}$ in equation (\ref{Deflectionangleazi}%
) are solved in closed analytic form as follows:%
\begin{eqnarray}
\pm \int_{\theta _{S}}^{\theta _{\mathrm{\min /\max }}} &=&\frac{\Phi }{2|a|}%
\sqrt[2]{\frac{(z_{m}-z_{S})}{z_{m}}}\frac{1}{\sqrt[2]{z_{m}-z_{3}}}\frac{2}{%
(1-z_{m})}F_{D}\left( \frac{1}{2},1,\frac{1}{2},\frac{1}{2},\frac{3}{2},%
\frac{z_{S}-z_{m}}{1-z_{m}},\frac{z_{m}-z_{S}}{z_{m}},\frac{z_{m}-z_{S}}{%
z_{m}-z_{3}}\right)  \notag \\
&&+[1-\mathrm{sign(}\theta _{S}\circ \theta _{mS}\mathrm{)]}\frac{\Phi }{|a|}%
\frac{z_{S}}{z_{m}}\frac{z_{S}-z_{m}}{1-z_{S}}\frac{1}{\sqrt[2]{%
z_{S}(z_{S}-z_{m})(z_{3}-z_{S})}}\times  \notag \\
&&F_{D}\left( 1,1,-\frac{1}{2},\frac{1}{2},\frac{3}{2},\frac{z_{S}(1-z_{m})}{%
z_{m}(1-z_{S})},\frac{z_{S}}{z_{m}},\frac{z_{S}(z_{m}-z_{3})}{%
z_{m}(z_{S}-z_{3)}}\right)  \label{Sourceangular}
\end{eqnarray}%
An equivalent expression for the above integral is:%
\begin{eqnarray}
\pm \int_{\theta _{S}}^{\theta _{\mathrm{\min /\max }}} &=&\frac{\Phi }{2|a|}%
\sqrt[2]{\frac{(z_{m}-z_{S})}{z_{m}}}\frac{1}{\sqrt[2]{z_{m}-z_{3}}}\frac{2}{%
(1-z_{m})}F_{D}\left( \frac{1}{2},1,\frac{1}{2},\frac{1}{2},\frac{3}{2},%
\frac{z_{S}-z_{m}}{1-z_{m}},\frac{z_{S}-z_{m}}{z_{m}},\frac{z_{m}-z_{S}}{%
z_{m}-z_{3}}\right)  \notag \\
&&+[1-\mathrm{sign(}\theta _{S}\circ \theta _{mS}\mathrm{)]}\frac{\Phi }{|a|}%
\sqrt[2]{\frac{z_{S}}{z_{m}}}\sqrt[2]{\frac{z_{m}-z_{3}}{z_{S}-z_{3}}}\frac{1%
}{\sqrt[2]{z_{m}-z_{3}}}\times  \notag \\
&&F_{D}\left( \frac{1}{2},1,\frac{1}{2},-\frac{1}{2},\frac{3}{2},\frac{%
z_{S}(1-z_{3})}{z_{S}-z_{3}},\frac{z_{S}}{z_{m}}\frac{(z_{m}-z_{3})}{%
(z_{S}-z_{3})},\frac{z_{S}}{z_{S}-z_{3}}\right)  \notag \\
&=&\frac{\Phi }{2|a|}\sqrt[2]{\frac{(z_{m}-z_{S})}{z_{m}}}\frac{1}{\sqrt[2]{%
z_{m}-z_{3}}}\frac{2}{(1-z_{m})}F_{D}\left( \frac{1}{2},1,\frac{1}{2},\frac{1%
}{2},\frac{3}{2},\frac{z_{S}-z_{m}}{1-z_{m}},\frac{z_{S}-z_{m}}{z_{m}},\frac{%
z_{m}-z_{S}}{z_{m}-z_{3}}\right)  \notag \\
&&+[1-\mathrm{sign(}\theta _{S}\circ \theta _{mS}\mathrm{)]}\frac{\Phi }{|a|}%
\sqrt[2]{\frac{z_{S}}{z_{m}}}\sqrt[2]{\frac{z_{m}-z_{3}}{z_{S}-z_{3}}}\frac{1%
}{\sqrt[2]{z_{m}-z_{3}}}\times  \notag \\
&&[\frac{-z_{3}}{1-z_{3}}F_{D}\left( \frac{1}{2},1,\frac{1}{2},\frac{1}{2},%
\frac{3}{2},\frac{z_{S}(1-z_{3})}{z_{S}-z_{3}},\frac{z_{S}}{z_{m}}\frac{%
(z_{m}-z_{3})}{(z_{S}-z_{3})},\frac{z_{S}}{z_{S}-z_{3}}\right) +  \notag \\
&&\frac{1}{1-z_{3}}F_{1}\left( \frac{1}{2},\frac{1}{2},\frac{1}{2},\frac{3}{2%
},\frac{z_{S}}{z_{m}}\frac{(z_{m}-z_{3})}{(z_{S}-z_{3})},\frac{z_{S}}{%
z_{S}-z_{3}}\right) ]  \label{gwniakimetasximaolokliro}
\end{eqnarray}

\bigskip In going from equation (\ref{Sourceangular}) to equation (\ref%
{gwniakimetasximaolokliro}) we made use of the functional equation of
Lauricella's hypergeometric function $F_{D}$ , Proposition 2 (\ref%
{Gwniakimatesxi2FD}), and Proposition 3 which are proved in Appendix \ref%
{Veta}.

Now, for a light trajectory that encounters $m$ turning points $(m\geq 1)$
in the polar motion we have\footnote{%
Recall the constraints of section \ref{Shadowonthewall}.}:

\begin{equation}
\fbox{$ \displaystyle \pm \int_{\theta_S}^{\theta_{\rm {min/max}}}
\underbrace{\pm \int_{\theta_{\rm {min/max}}}^{\theta_{\rm {max/min}}}
\pm \int_{\theta_{\rm {max/min}}}^{\theta_{\rm {min/max}}}\cdots}_{m-1\;{\rm times}}\pm
\int_{\theta_{\rm {max/min}}}^{\theta_O}= $}
\end{equation}%

\begin{eqnarray}
&=&\int_{z_{S}}^{z_{m}}+[1-\mathrm{sign(\theta }_{S}\circ \theta
_{mS})]\int_{0}^{z_{S}}  \notag \\
&&+\int_{z_{O}}^{z_{m}}+[1-\mathrm{sign(\theta }_{O}\circ \theta
_{mO})]\int_{0}^{z_{O}}  \notag \\
&&+2(m-1)\int_{0}^{z_{m}}
\end{eqnarray}%
where:

\begin{equation}
\fbox{$\displaystyle \theta_{mO}:={\rm Arccos}({\rm sign}(y_i)\sqrt{z_m})=
{\rm Arccos}({\rm sign}(\beta_i)\sqrt{z_m}), $}
\end{equation}
$y_{i}$ is the possible position of the image and: 
\begin{equation}
\theta _{mS}:=\left\{ 
\begin{array}{lll}
\theta _{mO}, & m & \mathrm{odd} \\ 
\pi -\theta _{mO}, & m & \mathrm{even}%
\end{array}%
\right.
\end{equation}

Thus we have that :%
\begin{eqnarray}
A_{2}(x_{i},y_{i},x_{S,}y_{S},m) &=&2(m-1)\times \left[ \frac{\Phi }{|a|}%
\frac{1}{\sqrt[2]{(z_{m}-z_{3})}}\frac{1}{(1-z_{m})}\frac{\pi }{2}%
F_{1}\left( \frac{1}{2},1,\frac{1}{2},1,\frac{-z_{m}}{1-z_{m}},\frac{z_{m}}{%
z_{m}-z_{3}}\right) \right]  \notag \\
&&+\frac{\Phi }{2|a|}\sqrt[2]{\frac{(z_{m}-z_{S})}{z_{m}}}\frac{1}{\sqrt[2]{%
z_{m}-z_{3}}}\frac{2}{(1-z_{m})}\times  \notag \\
&&F_{D}\left( \frac{1}{2},1,\frac{1}{2},\frac{1}{2},\frac{3}{2},\frac{%
z_{S}-z_{m}}{1-z_{m}},\frac{z_{m}-z_{S}}{z_{m}},\frac{z_{m}-z_{S}}{%
z_{m}-z_{3}}\right)  \notag \\
&&+[1-\mathrm{sign(}\theta _{S}\circ \theta _{mS}\mathrm{)]}\frac{\Phi }{|a|}%
\frac{z_{S}}{z_{m}}\frac{z_{S}-z_{m}}{1-z_{S}}\frac{1}{\sqrt[2]{%
z_{S}(z_{S}-z_{m})(z_{3}-z_{S})}}\times  \notag \\
&&F_{D}\left( 1,1,-\frac{1}{2},\frac{1}{2},\frac{3}{2},\frac{z_{S}(1-z_{m})}{%
z_{m}(1-z_{S})},\frac{z_{S}}{z_{m}},\frac{z_{S}(z_{m}-z_{3})}{%
z_{m}(z_{S}-z_{3)}}\right) +  \notag \\
&&+\frac{\Phi }{2|a|}\sqrt[2]{\frac{(z_{m}-z_{O})}{z_{m}}}\frac{1}{\sqrt[2]{%
z_{m}-z_{3}}}\frac{2}{(1-z_{m})}\times  \notag \\
&&F_{D}\left( \frac{1}{2},1,\frac{1}{2},\frac{1}{2},\frac{3}{2},\frac{%
z_{O}-z_{m}}{1-z_{m}},\frac{z_{m}-z_{O}}{z_{m}},\frac{z_{m}-z_{O}}{%
z_{m}-z_{3}}\right)  \notag \\
&&[1-\mathrm{sign(}\theta _{O}\circ \theta _{mO}\mathrm{)]}\frac{\Phi }{|a|}%
\frac{z_{O}}{z_{m}}\frac{z_{O}-z_{m}}{1-z_{O}}\frac{1}{\sqrt[2]{%
z_{O}(z_{O}-z_{m})(z_{3}-z_{O})}}\times  \notag \\
&&F_{D}\left( 1,1,-\frac{1}{2},\frac{1}{2},\frac{3}{2},\frac{z_{O}(1-z_{m})}{%
z_{m}(1-z_{O})},\frac{z_{O}}{z_{m}},\frac{z_{O}(z_{m}-z_{3})}{%
z_{m}(z_{O}-z_{3)}}\right)  \notag \\
&&  \label{KleistiGwniaki2}
\end{eqnarray}

\bigskip We now calculate in closed form the angular term $%
A_{1}(x_{i},y_{i},x_{S},y_{S},m)$ which appears in equations (\ref%
{Lenseksiswsi}),(\ref{AbelIntegral}).

Indeed, the angular integrals of the form $\pm \int_{\theta _{S}}^{\theta _{%
\mathrm{\min /\max }}},$ in equation (\ref{AbelIntegral}), are computed in
closed-analytic form in terms of Appell's generalized hypergeometric
function of two variables as follows:

\begin{eqnarray}
\pm \int_{\theta _{S}}^{\theta _{\mathrm{\min /\max }}}\frac{\mathrm{d}%
\theta }{\sqrt[2]{\Theta }} &=&\frac{1}{2|a|}\frac{\sqrt[2]{(z_{m}-z_{S})}}{%
\sqrt[2]{z_{m}(z_{m}-z_{3})}}F_{1}\left( \frac{1}{2},\frac{1}{2},\frac{1}{2},%
\frac{3}{2},\frac{z_{m}-z_{S}}{z_{m}},\frac{z_{m}-z_{S}}{z_{m}-z_{3}}\right) 
\frac{\Gamma (\frac{1}{2})\Gamma (1)}{\Gamma (3/2)}  \notag \\
&&+[1-\mathrm{sign(}\theta _{s}\circ \theta _{ms})]\frac{1}{|a|}\frac{\sqrt[2%
]{\frac{z_{S}(z_{m}-z_{3})}{z_{m}(z_{S}-z_{3})}}}{\sqrt[2]{z_{m}-z_{3}}}%
\times  \notag \\
&&F_{1}\left( \frac{1}{2},\frac{1}{2},\frac{1}{2},\frac{3}{2},\frac{z_{m}}{%
z_{m}-z_{3}}\frac{z_{S}(z_{m}-z_{3})}{z_{m}(z_{S}-z_{3})},\frac{%
z_{S}(z_{m}-z_{3})}{z_{m}(z_{S}-z_{3})}\right)  \notag \\
&&  \label{SimplerMarker}
\end{eqnarray}

\bigskip Also the integral (\ref{Firstang}) is calculated for $z_{j}=0$ in
terms of ordinary Gau\ss 's hypergeometric function:%
\begin{eqnarray}
&&2(m-1)\int_{0}^{z_{m}}\frac{\mathrm{d}z}{\sqrt[2]{z(z_{m}-z)(z-z_{3})}} \\
&=&\frac{2(m-1)}{2|a|}\sqrt[2]{\frac{z_{m}}{z_{m}(z_{m}-z_{3})}}\pi F\left( 
\frac{1}{2},\frac{1}{2},1,\frac{z_{m}}{z_{m}-z_{3}}\right)
\end{eqnarray}

Thus we obtain, 
\begin{eqnarray}
A_{1}(x_{i},y_{i},x_{S,}y_{S},m) &=&2(m-1)\frac{1}{2|a|}\sqrt{\frac{z_{m}}{%
z_{m}(z_{m}-z_{3})}}\pi F\left( \frac{1}{2},\frac{1}{2},1,\frac{z_{m}}{%
z_{m}-z_{3}}\right) +  \notag \\
&&\frac{1}{2|a|}\frac{\sqrt[2]{(z_{m}-z_{S})}}{\sqrt[2]{z_{m}(z_{m}-z_{3})}}%
F_{1}\left( \frac{1}{2},\frac{1}{2},\frac{1}{2},\frac{3}{2},\frac{z_{m}-z_{S}%
}{z_{m}},\frac{z_{m}-z_{S}}{z_{m}-z_{3}}\right) \frac{\Gamma (\frac{1}{2}%
)\Gamma (1)}{\Gamma (3/2)}  \notag \\
&&+[1-\mathrm{sign(}\theta _{S}\circ \theta _{mS})]\frac{1}{|a|}\frac{\sqrt[2%
]{\frac{z_{S}(z_{m}-z_{3})}{z_{m}(z_{S}-z_{3})}}}{\sqrt[2]{z_{m}-z_{3}}}%
\times  \notag \\
&&F_{1}\left( \frac{1}{2},\frac{1}{2},\frac{1}{2},\frac{3}{2},\frac{z_{m}}{%
z_{m}-z_{3}}\frac{z_{S}(z_{m}-z_{3})}{z_{m}(z_{S}-z_{3})},\frac{%
z_{S}(z_{m}-z_{3})}{z_{m}(z_{S}-z_{3})}\right) +  \notag \\
&&\frac{1}{2|a|}\frac{\sqrt[2]{(z_{m}-z_{O})}}{\sqrt[2]{z_{m}(z_{m}-z_{3})}}%
F_{1}\left( \frac{1}{2},\frac{1}{2},\frac{1}{2},\frac{3}{2},\frac{z_{m}-z_{O}%
}{z_{m}},\frac{z_{m}-z_{O}}{z_{m}-z_{3}}\right) \frac{\Gamma (\frac{1}{2}%
)\Gamma (1)}{\Gamma (3/2)}  \notag \\
&&+[1-\mathrm{sign(}\theta _{O}\circ \theta _{mO})]\frac{1}{|a|}\frac{\sqrt[2%
]{\frac{z_{O}(z_{m}-z_{3})}{z_{m}(z_{O}-z_{3})}}}{\sqrt[2]{z_{m}-z_{3}}}%
\times  \notag \\
&&F_{1}\left( \frac{1}{2},\frac{1}{2},\frac{1}{2},\frac{3}{2},\frac{z_{m}}{%
z_{m}-z_{3}}\frac{z_{O}(z_{m}-z_{3})}{z_{m}(z_{O}-z_{3})},\frac{%
z_{O}(z_{m}-z_{3})}{z_{m}(z_{O}-z_{3})}\right)  \label{GvniaKleisti1}
\end{eqnarray}

\bigskip For $m=0$ i.e. for no turning points in the polar coordinate the
exact solutions for the angular integrals in equation (\ref{AbelIntegral}), (%
\ref{Deflectionangleazi}) become

\begin{eqnarray}
A_{1}(x_{i},y_{i},x_{S},y_{S}) &=&\pm \int_{\theta _{S}}^{\theta
_{O}}=\int_{z_{1}}^{z_{2}}+(1-\mathrm{sign(}\theta _{S}\circ \theta
_{O}))\int_{0}^{z_{1}}  \notag \\
&=&\int_{z_{1}}^{z_{m}}-\int_{z_{2}}^{z_{m}}+(1-\mathrm{sign}(\theta
_{S}\circ \theta _{O}))\int_{0}^{z_{1}}  \notag \\
&=&\frac{1}{2|a|}\frac{\sqrt[2]{z_{m}-z_{1}}}{\sqrt[2]{z_{m}(z_{m}-z_{3})}}%
F_{1}\left( \frac{1}{2},\frac{1}{2},\frac{1}{2},\frac{3}{2},\frac{z_{m}-z_{1}%
}{z_{m}},\frac{z_{m}-z_{1}}{z_{m}-z_{3}}\right) 2  \notag \\
&&-\frac{1}{2|a|}\frac{\sqrt[2]{z_{m}-z_{2}}}{\sqrt[2]{z_{m}(z_{m}-z_{3})}}%
F_{1}\left( \frac{1}{2},\frac{1}{2},\frac{1}{2},\frac{3}{2},\frac{z_{m}-z_{2}%
}{z_{m}},\frac{z_{m}-z_{2}}{z_{m}-z_{3}}\right) 2  \notag \\
&&+[1-\mathrm{sign}(\theta _{S}\circ \theta _{O})]\frac{1}{|a|}\frac{\sqrt[2]%
{\frac{z_{1}(z_{m}-z_{3})}{z_{m}(z_{1}-z_{3)}}}}{\sqrt[2]{z_{m}-z_{3}}} 
\notag \\
&&\times F_{1}\left( \frac{1}{2},\frac{1}{2},\frac{1}{2},\frac{3}{2},\frac{%
z_{m}}{z_{m}-z_{3}}\frac{z_{1}(z_{m}-z_{3})}{z_{m}(z_{1}-z_{3})},\frac{%
z_{1}(z_{m}-z_{3})}{z_{m}(z_{1}-z_{3})}\right)  \notag \\
&&  \label{MZeroA1}
\end{eqnarray}%
and 
\begin{eqnarray}
A_{2}(x_{i},y_{i},x_{S},y_{S}) &=&\frac{\Phi }{2|a|}\sqrt{\frac{(z_{m}-z_{1})%
}{z_{m}}}\frac{1}{\sqrt{z_{m}-z_{3}}}\frac{2}{1-z_{m}}  \notag \\
&&\times F_{D}\left( \frac{1}{2},1,\frac{1}{2},\frac{1}{2},\frac{3}{2},\frac{%
z_{1}-z_{m}}{1-z_{m}},\frac{z_{m}-z_{1}}{z_{m}},\frac{z_{m}-z_{1}}{%
z_{m}-z_{3}}\right)  \notag \\
&&-%
\Biggl[%
\frac{\Phi }{2|a|}\sqrt{\frac{(z_{m}-z_{2})}{z_{m}}}\frac{1}{\sqrt{%
z_{m}-z_{3}}}\frac{2}{1-z_{m}}  \notag \\
&&\times F_{D}\left( \frac{1}{2},1,\frac{1}{2},\frac{1}{2},\frac{3}{2},\frac{%
z_{2}-z_{m}}{1-z_{m}},\frac{z_{m}-z_{2}}{z_{m}},\frac{z_{m}-z_{2}}{%
z_{m}-z_{3}}\right) 
\Biggr]
\notag \\
&&+[1-\mathrm{sign(}\theta _{S}\circ \theta _{O}\mathrm{)]}\frac{\Phi }{|a|}%
\sqrt{\frac{z_{1}}{z_{m}}}\sqrt{\frac{z_{m}-z_{3}}{z_{1}-z_{3}}}\frac{1}{%
\sqrt{z_{m}-z_{3}}}  \notag \\
&&\times F_{D}\left( \frac{1}{2},1,\frac{1}{2},-\frac{1}{2},\frac{3}{2},%
\frac{z_{1}(1-z_{3})}{z_{1}-z_{3}},\frac{z_{1}}{z_{m}}\frac{z_{m}-z_{3}}{%
z_{1}-z_{3}},\frac{z_{1}}{z_{1}-z_{3}}\right)  \notag \\
&&  \label{A2Mzero}
\end{eqnarray}%
where $z_{1}:=\min (z_{S},z_{O}),$ $z_{2}:=\max (z_{S},z_{O}).$

Equations (\ref{GvniaKleisti1}),(\ref{KleistiGwniaki2}) for $m\geq 1$
turning points and (\ref{MZeroA1}),(\ref{A2Mzero}) for $m=0$ turning points
constitute our exact results for the angular integrals which appear in (\ref%
{Lenseksiswsi}) for the case of vanishing cosmological constant $\Lambda $.
\ It is time to turn our attention to the exact computation of the radial
integrals which appear in the lens-equations of the Kerr black hole.

\section{Closed form solution for the radial integrals\label{AktiOlok}.}

We now perform the radial integration assuming $\Lambda =0.$

For an observer and a source located far away from the black hole, the
relevant radial integrals can take the form:

\begin{equation}
\fbox{$\displaystyle \int^r\rightarrow -\int_{r_S}^{\alpha}+\int_{\alpha}^{r_O}\simeq2\int_{\alpha}^{\infty}$}
\end{equation}%

For instance, in the calculation of the azimuthial coordinate (\ref%
{Deflectionangleazi}) the following radial integral is involved:

\begin{equation}
\int_{\alpha }^{\infty }\frac{a}{\Delta }[(r^{2}+a^{2})-a\Phi ]\frac{%
\mathrm{d}r}{\sqrt[2]{R}}
\end{equation}%
where $\Delta :=r^{2}+a^{2}-2GM\frac{r}{c^{2}}$ . In order to calculate the
contribution to the deflection angle from the radial term we need to
integrate the above equation from the 
{\color{blue}{distance\;of\;closest\;approach}}
(e.g., from the maximum positive root of the quartic) to infinity. We denote
the roots of the quartic polynomial $R$ (eqn (\ref{quartic1}) for $\Lambda
=0)$ by $\alpha ,\beta ,\gamma ,\delta :\alpha >\beta >\gamma >\delta .$ We
manipulate first the terms:

\begin{equation}
\fbox{$\displaystyle
\int_{\alpha}^{\infty}\frac{a}{\Delta}\frac{(r^2+a^2)}{\sqrt{R}}{\rm d}r=
\int_{\alpha}^{\infty}\frac{a{\rm d}r}{\sqrt{R}}
\Biggl[1+\frac{\frac{2GM}{c^2}r}{\underbrace{r^2+a^2-\frac{2GM}{c^2}r}_{\Delta}}\Biggr]=
\int_{\alpha}^{\infty}\frac{a{\rm d}r}{\sqrt{R}}+
\int_{\alpha}^{\infty}\frac{a\frac{2GMr}{c^2}{\rm d}r}{\Delta\sqrt{R}} $}
\end{equation}%

It is enough to proceed with the term \footnote{%
The radial term $2\int_{\alpha }^{\infty }\frac{a\mathrm{d}r}{\sqrt{R}}$ is
cancelled from the angular term: $-\int \frac{a\mathrm{d}\theta }{\sqrt{%
\Theta }}.$}:

\begin{equation}
\int_{\alpha }^{\infty }\frac{a\frac{2GM}{c^{2}}r-a^{2}\Phi }{\Delta 
\sqrt[2]{R}}\mathrm{d}r
\end{equation}

Expressing the roots of $\Delta $ as $r_{+},r_{-},$ which are the radii of
the event horizon and the inner or Cauchy horizon respectively, and using
partial fractions we derive the expression:

\begin{eqnarray}
\int_{\alpha }^{\infty }\frac{a\frac{2GM}{c^{2}}r-a^{2}\Phi }{\Delta 
\sqrt[2]{R}}\mathrm{d}r &=&\int_{\alpha }^{\infty }\frac{A_{+}^{go}}{%
(r-r_{+})\sqrt[2]{R}}\mathrm{d}r+\int_{\alpha }^{\infty }\frac{%
A_{-}^{go}}{(r-r_{-})\sqrt[2]{R}}\mathrm{d}r  \notag \\
&=&\int_{\alpha }^{\infty }\frac{A_{+}^{go}}{(r-r_{+})\sqrt[2]{%
(r-\alpha )(r-\beta )(r-\gamma )(r-\delta )}}\mathrm{d}r  \notag \\
&&+\int_{\alpha }^{\infty }\frac{A_{-}^{go}}{(r-r_{-})\sqrt[2]{%
(r-\alpha )(r-\beta )(r-\gamma )(r-\delta )}}\mathrm{d}r  \notag \\
&&
\end{eqnarray}%
where $A_{\pm }^{go}$ are given by the equations

\begin{equation}
A_{\pm }^{go}=\pm \frac{(r_{\pm }a2\frac{GM}{c^{2}}-a^{2}\Phi )}{r_{+}-r_{-}}
\label{gorbitcoef}
\end{equation}

For polar orbits $\Phi =0$ and the coefficients in (\ref{gorbitcoef}) reduce
to those calculated in \cite{Kraniotis}.

We organize all roots in ascending order of magnitude as follows\footnote{%
We have the correspondence $\alpha _{\mu +1}=\alpha ,\alpha _{\mu +2}=\beta
,\alpha _{\mu -1}=r_{+}=\alpha _{\mu -2},\alpha _{\mu -3}=\gamma ,\alpha
_{\mu }=\delta .$},

\begin{equation}
\alpha _{\mu }>\alpha _{\nu }>\alpha _{i}>\alpha _{\rho }
\end{equation}%
where $\alpha _{\mu }=\alpha _{\mu +1},\alpha _{\nu }=\alpha _{\mu
+2},\alpha _{\rho }=\alpha _{\mu }$ and $\alpha _{i}=\alpha _{\mu
-i},i=1,2,3 $ and we have that $\alpha _{\mu -1}\geq \alpha _{\mu -2}>\alpha
_{\mu -3}.$ By applying the transformation

\begin{equation}
r=\frac{\omega z\alpha _{\mu +2}-\alpha _{\mu +1}}{\omega z-1}
\end{equation}%
or equivalently

\begin{equation}
z=\left( \frac{\alpha _{\mu }-\alpha _{\mu +2}}{\alpha _{\mu }-\alpha _{\mu
+1}}\right) \left( \frac{r-\alpha _{\mu +1}}{r-\alpha _{\mu +2}}\right)
\end{equation}%
where

\begin{equation}
\omega :=\frac{\alpha _{\mu }-\alpha _{\mu +1}}{\alpha _{\mu }-\alpha _{\mu
+2}}
\end{equation}%
we can bring our radial integrals into the familiar integral representation
of Lauricella's $F_{D}$ and Appell's hypergeometric function $F_{1\text{ \ }%
} $of three and two variables respectively. Indeed, we derive

\begin{eqnarray}
\Delta \phi _{r_{1}}^{go} &=&2%
\Biggl[%
\int_{0}^{1/\omega }\frac{-A_{+}^{go}\omega (\alpha _{\mu +1}-\alpha
_{\mu +2})}{H^{+}}\frac{\mathrm{d}z}{\sqrt[2]{z(1-z)}(1-\kappa _{+}^{2}z)%
\sqrt[2]{1-\mu ^{2}z}}  \notag \\
&&+\int_{0}^{1/\omega }\frac{A_{+}^{go}\omega ^{2}(\alpha _{\mu
+1}-\alpha _{\mu +2})}{H^{+}}\frac{z\mathrm{d}z}{\sqrt[2]{z(1-z)}(1-\kappa
_{+}^{2}z)\sqrt[2]{1-\mu ^{2}z}}  \notag \\
&&+\int_{0}^{1/\omega }\frac{-A_{-}^{go}\omega (\alpha _{\mu
+1}-\alpha _{\mu +2})}{H^{-}}\frac{\mathrm{d}z}{\sqrt[2]{z(1-z)}(1-\kappa
_{-}^{2}z)\sqrt[2]{1-\mu ^{2}z}}  \notag \\
&&+\int_{0}^{1/\omega }\frac{A_{-}^{go}\omega ^{2}(\alpha _{\mu
+1}-\alpha _{\mu +2})}{H^{-}}\frac{z\mathrm{d}z}{\sqrt[2]{z(1-z)}(1-\kappa
_{-}^{2}z)\sqrt[2]{1-\mu ^{2}z}}%
\Biggr]%
\end{eqnarray}%
where the moduli $\kappa _{\pm }^{2},\mu ^{2}$ are

\begin{equation}
\kappa _{\pm }^{2}=\left( \frac{\alpha _{\mu }-\alpha _{\mu +1}}{\alpha
_{\mu }-\alpha _{\mu +2}}\right) \left( \frac{\alpha _{\mu +2}-\alpha _{\mu
-1}^{\pm }}{\alpha _{\mu +1}-\alpha _{\mu -1}^{\pm }}\right) ,\qquad \mu
^{2}=\left( \frac{\alpha _{\mu }-\alpha _{\mu +1}}{\alpha _{\mu }-\alpha
_{\mu +2}}\right) \left( \frac{\alpha _{\mu +2}-\alpha _{\mu -3}}{\alpha
_{\mu +1}-\alpha _{\mu -3}}\right)
\end{equation}

Also

\begin{equation}
H^{\pm }=\sqrt[2]{\omega }(\alpha _{\mu +1}-\alpha _{\mu +2})(\alpha _{\mu
+1}-\alpha _{\mu -1}^{\pm })\sqrt[2]{\alpha _{\mu +1}-\alpha _{\mu }}\sqrt[2]%
{\alpha _{\mu +1}-\alpha _{\mu -3}}
\end{equation}%
and $\alpha _{\mu -1}^{\pm }=r_{\pm }.$ By defining a new variable $%
z^{\prime }:=\omega z$ we can express the contribution $\Delta \phi
_{r_{1}}^{go}$ ,to the deflection angle, from the above radial terms in
terms of Lauricella's hypergeometric function $F_{D}$

\begin{eqnarray}
\Delta \phi _{r_{1}}^{go} &=&2%
\Biggl[%
\frac{-2A_{+}^{go}\sqrt{\omega }(\alpha _{\mu +1}-\alpha _{\mu +2})}{H^{+}}%
F_{D}\left( \frac{1}{2},\frac{1}{2},1,\frac{1}{2},\frac{3}{2},\frac{1}{%
\omega },\kappa _{+}^{\prime 2},\mu ^{\prime 2}\right)  \notag \\
&&+\frac{A_{+}^{go}\sqrt{\omega }(\alpha _{\mu +1}-\alpha _{\mu +2})}{H^{+}}%
F_{D}\left( \frac{3}{2},\frac{1}{2},1,\frac{1}{2},\frac{5}{2},\frac{1}{%
\omega },\kappa _{+}^{\prime 2},\mu ^{\prime 2}\right) \frac{\Gamma
(3/2)\Gamma (1)}{\Gamma (5/2)}  \notag \\
&&+\frac{-2A_{-}^{go}\sqrt{\omega }(\alpha _{\mu +1}-\alpha _{\mu +2})}{H^{-}%
}F_{D}\left( \frac{1}{2},\frac{1}{2},1,\frac{1}{2},\frac{3}{2},\frac{1}{%
\omega },\kappa _{-}^{\prime 2},\mu ^{\prime 2}\right)  \notag \\
&&+\frac{A_{-}^{go}\sqrt{\omega }(\alpha _{\mu +1}-\alpha _{\mu +2})}{H^{-}}%
F_{D}\left( \frac{3}{2},\frac{1}{2},1,\frac{1}{2},\frac{5}{2},\frac{1}{%
\omega },\kappa _{-}^{\prime 2},\mu ^{\prime 2}\right) \frac{\Gamma
(3/2)\Gamma (1)}{\Gamma (5/2)}%
\Biggr]
\notag  \label{ArxiLaurice} \\
&&
\end{eqnarray}%
where the variables of the function $F_{D}$ are given in terms of the roots
of the quartic and the radii of the event and Cauchy horizons by the
expressions

\begin{eqnarray}
\frac{1}{\omega } &=&\frac{\alpha _{\mu }-\alpha _{\mu +2}}{\alpha _{\mu
}-\alpha _{\mu +1}}=\frac{\delta -\beta }{\delta -\alpha }  \notag \\
\kappa _{\pm }^{\prime 2} &=&\frac{\alpha _{\mu +2}-\alpha _{\mu -1}^{\pm }}{%
\alpha _{\mu +1}-\alpha _{\mu -1}^{\pm }}=\frac{\beta -r_{\pm }}{\alpha
-r_{\pm }} \\
\mu ^{\prime 2} &=&\frac{\alpha _{\mu +2}-\alpha _{\mu -3}}{\alpha _{\mu
+1}-\alpha _{\mu -3}}=\frac{\beta -\gamma }{\alpha -\gamma }  \notag
\end{eqnarray}

An equivalent expression is as follows 
\begin{eqnarray}
\Delta \phi _{r_{1}}^{go} &=&2%
\Biggl[%
\frac{-2A_{+}^{go}\sqrt{\omega }(\alpha _{\mu +1}-\alpha _{\mu +2})}{H^{+}}%
F_{D}\left( \frac{1}{2},\frac{1}{2},1,\frac{1}{2},\frac{3}{2},\frac{1}{%
\omega },\kappa _{+}^{\prime 2},\mu ^{\prime 2}\right)  \notag \\
&&+\frac{A_{+}^{go}\sqrt{\omega }(\alpha _{\mu +1}-\alpha _{\mu +2})}{H^{+}}%
\Biggl(%
-\frac{1}{\kappa _{+}^{\prime 2}}F_{1}\left( \frac{1}{2},\frac{1}{2},\frac{1%
}{2},\frac{3}{2},\frac{1}{\omega },\mu ^{\prime 2}\right) 2  \notag \\
&&+\frac{1}{\kappa _{+}^{\prime 2}}F_{D}\left( \frac{1}{2},\frac{1}{2},1,%
\frac{1}{2},\frac{3}{2},\frac{1}{\omega },\kappa _{+}^{\prime 2},\mu
^{\prime 2}\right) 2%
\Biggr)
\notag \\
&&+\frac{-2A_{-}^{go}\sqrt{\omega }(\alpha _{\mu +1}-\alpha _{\mu +2})}{H^{-}%
}F_{D}\left( \frac{1}{2},\frac{1}{2},1,\frac{1}{2},\frac{3}{2},\frac{1}{%
\omega },\kappa _{-}^{\prime 2},\mu ^{\prime 2}\right)  \notag \\
&&+\frac{A_{-}^{go}\sqrt{\omega }(\alpha _{\mu +1}-\alpha _{\mu +2})}{H^{-}}%
\Biggl(%
-\frac{1}{\kappa _{-}^{\prime 2}}F_{1}\left( \frac{1}{2},\frac{1}{2},\frac{1%
}{2},\frac{3}{2},\frac{1}{\omega },\mu ^{\prime 2}\right) 2  \notag \\
&&+\frac{1}{\kappa _{-}^{\prime 2}}F_{D}\left( \frac{1}{2},\frac{1}{2},1,%
\frac{1}{2},\frac{3}{2},\frac{1}{\omega },\kappa _{-}^{\prime 2},\mu
^{\prime 2}\right) 2%
\Biggr)%
\Biggr]
\notag \\
&\equiv &R_{2}(x_{i},y_{i})  \label{Lauricella4}
\end{eqnarray}

In going from (\ref{ArxiLaurice}) to (\ref{Lauricella4}) we used the
identity proven in \cite{Kraniotis}, eqn.(52) in \cite{Kraniotis}.

Finally, the term $\int_{\alpha }^{\infty }\frac{\mathrm{d}r}{\sqrt{R%
}}$ ,is calculated in closed form in terms of Appell's first hypergeometric
function of two-variables :

\begin{eqnarray}
\int_{\alpha }^{\infty }\frac{\mathrm{d}r}{\sqrt{R}} &=&\frac{%
(\alpha _{\mu +1}-\alpha _{\mu +2})}{\sqrt[2]{(\alpha _{\mu +2}-\alpha _{\mu
+1})^{2}(\alpha _{\mu -1}-\alpha _{\mu +1})(\alpha _{\mu }-\alpha _{\mu +1})}%
}\times  \notag \\
&&\frac{\Gamma (1/2)}{\Gamma (3/2)}F_{1}\left( \frac{1}{2},\frac{1}{2},\frac{%
1}{2},\frac{3}{2},\lambda _{1}^{2},\frac{1}{\omega }\right)  \notag \\
&=&\frac{1}{\sqrt{(\alpha -\gamma )(\alpha -\delta )}}\frac{\Gamma (1/2)}{%
\Gamma (3/2)}F_{1}\left( \frac{1}{2},\frac{1}{2},\frac{1}{2},\frac{3}{2},%
\frac{\beta -\gamma }{\alpha -\gamma },\frac{\delta -\beta }{\delta -\alpha }%
\right)  \label{R1ROne}
\end{eqnarray}

We exploit further the lens-equations (\ref{Lenseksiswsi}). Indeed:%
\begin{eqnarray}
R_{1}(x_{i},y_{i})-2(m-1)\frac{1}{2|a|}\sqrt{\frac{z_{m}}{z_{m}(z_{m}-z_{3})}%
}\pi F\left( \frac{1}{2},\frac{1}{2},1,\frac{z_{m}}{z_{m}-z_{3}}\right)
+\cdots &=&\int^{\xi _{S}}\frac{\mathrm{d\xi }}{\sqrt{4\xi ^{3}-g_{2}\xi
-g_{3}}}  \notag \\
&&
\end{eqnarray}%
Inverting:

\begin{equation}
\xi _{S}=\wp \left( \frac{2(\ref{R1ROne})}{1}-2(m-1)\frac{1}{2|a|}\sqrt{%
\frac{z_{m}}{z_{m}(z_{m}-z_{3})}}\pi F\left( \frac{1}{2},\frac{1}{2},1,\frac{%
z_{m}}{z_{m}-z_{3}}\right) +\cdots +\epsilon \right)  \label{WeierstrassKarl}
\end{equation}%
while:%
\begin{equation}
-\phi _{S}=R_{2}(x_{i},y_{i})+A_{2}(x_{i},y_{i},x_{S,}y_{S},m)
\end{equation}%
where $\wp (z)$ denotes the Weierstra\ss\ elliptic\ function (which is also
a meromorphic Jacobi modular form of weight 2 ) and the Weierstra\ss\ %
invariants are given in terms of the initial conditions by:%
\begin{eqnarray}
g_{2} &=&\frac{1}{12}(\alpha +\beta )^{2}-\mathcal{Q}\frac{\alpha }{4}, \\
g_{3} &=&\frac{1}{216}(\alpha +\beta )^{3}-\mathcal{Q}\frac{\alpha ^{2}}{48}-%
\mathcal{Q}\frac{\alpha \beta }{48}
\end{eqnarray}%
Also $\alpha :=-a^{2},$ $\beta :=\mathcal{Q+}\Phi ^{2},$ $z_{S}=-\frac{\xi
_{S}+\frac{\alpha +\beta }{12}}{-\alpha /4}$ and $\epsilon $ is a constant
of intergration.

\bigskip To recapitulate our exact solutions of the lens equations (\ref%
{Lenseksiswsi}) are given by:%
\begin{eqnarray}
2\int_{\alpha }^{\infty }\frac{1}{\sqrt{R}}\mathrm{d}r
&=&A_{1}(x_{i},y_{i},x_{S},y_{S},m)\Leftrightarrow \frac{2}{\sqrt{(\alpha
-\gamma )(\alpha -\delta )}}\frac{\Gamma (1/2)}{\Gamma (3/2)}F_{1}\left( 
\frac{1}{2},\frac{1}{2},\frac{1}{2},\frac{3}{2},\frac{\beta -\gamma }{\alpha
-\gamma },\frac{\delta -\beta }{\delta -\alpha }\right) =  \notag \\
&&2(m-1)\frac{1}{2|a|}\sqrt{\frac{z_{m}}{z_{m}(z_{m}-z_{3})}}\pi F\left( 
\frac{1}{2},\frac{1}{2},1,\frac{z_{m}}{z_{m}-z_{3}}\right) +  \notag \\
&&\frac{1}{2|a|}\frac{\sqrt[2]{(z_{m}-z_{S})}}{\sqrt[2]{z_{m}(z_{m}-z_{3})}}%
F_{1}\left( \frac{1}{2},\frac{1}{2},\frac{1}{2},\frac{3}{2},\frac{z_{m}-z_{S}%
}{z_{m}},\frac{z_{m}-z_{S}}{z_{m}-z_{3}}\right) \frac{\Gamma (\frac{1}{2}%
)\Gamma (1)}{\Gamma (3/2)}  \notag \\
&&+[1-\mathrm{sign(}\theta _{S}\circ \theta _{mS})]\frac{1}{|a|}\frac{\sqrt[2%
]{\frac{z_{S}(z_{m}-z_{3})}{z_{m}(z_{S}-z_{3})}}}{\sqrt[2]{z_{m}-z_{3}}}%
\times  \notag \\
&&F_{1}\left( \frac{1}{2},\frac{1}{2},\frac{1}{2},\frac{3}{2},\frac{z_{m}}{%
z_{m}-z_{3}}\frac{z_{S}(z_{m}-z_{3})}{z_{m}(z_{S}-z_{3})},\frac{%
z_{S}(z_{m}-z_{3})}{z_{m}(z_{S}-z_{3})}\right) +  \notag \\
&&\frac{1}{2|a|}\frac{\sqrt[2]{(z_{m}-z_{O})}}{\sqrt[2]{z_{m}(z_{m}-z_{3})}}%
F_{1}\left( \frac{1}{2},\frac{1}{2},\frac{1}{2},\frac{3}{2},\frac{z_{m}-z_{O}%
}{z_{m}},\frac{z_{m}-z_{O}}{z_{m}-z_{3}}\right) \frac{\Gamma (\frac{1}{2}%
)\Gamma (1)}{\Gamma (3/2)}  \notag \\
&&+[1-\mathrm{sign(}\theta _{O}\circ \theta _{mO})]\frac{1}{|a|}\frac{\sqrt[2%
]{\frac{z_{O}(z_{m}-z_{3})}{z_{m}(z_{O}-z_{3})}}}{\sqrt[2]{z_{m}-z_{3}}}%
\times  \notag \\
&&F_{1}\left( \frac{1}{2},\frac{1}{2},\frac{1}{2},\frac{3}{2},\frac{z_{m}}{%
z_{m}-z_{3}}\frac{z_{O}(z_{m}-z_{3})}{z_{m}(z_{O}-z_{3})},\frac{%
z_{O}(z_{m}-z_{3})}{z_{m}(z_{O}-z_{3})}\right) ,  \label{Constraint1one}
\end{eqnarray}

\begin{equation}
-\phi _{S}=R_{2}(x_{i},y_{i})+A_{2}(x_{i},y_{i},x_{S},y_{S},m)
\label{Constraint2}
\end{equation}%
and equation (\ref{WeierstrassKarl}).

In the subsequent sections we shall apply our exact solutions for the lens
equations in the Kerr geometry expressed by (\ref{Constraint1one}),(\ref%
{WeierstrassKarl}),(\ref{Constraint2}) to particular cases which include: a)
an equatorial observer: $\theta _{O}=\pi /2\Rightarrow z_{O}=0$ b) a generic
observer located at $\theta _{O}=\pi /3\Rightarrow z_{O}=\frac{1}{4}$.

\section{Positions of images, source and resulting magnifications for an
equatorial observer in a Kerr black hole\label{IsimeParatiritis}.}

In this case ($\theta _{O}=\pi /2$), equations (\ref{ObserSkym}),(\ref%
{Integr}), become:%
\begin{eqnarray}
\Phi &\simeq &-\alpha _{i}\sin \theta _{O}=-\alpha _{i}  \notag \\
\mathcal{Q} &\simeq &\beta _{i}^{2}+(\alpha _{i}^{2}-a^{2})\cos ^{2}\theta
_{O}=\beta _{i}^{2}  \label{Isimerinieikona}
\end{eqnarray}%
Thus the length of the vector on the observer's image plane equals to:%
\begin{equation}
\sqrt{\alpha _{i}^{2}+\beta _{i}^{2}}=\sqrt{\Phi ^{2}+\mathcal{Q}}
\end{equation}

Furthermore, we derive the equations:

\begin{equation}
\fbox{$\displaystyle x_S:=\frac{\alpha_S}{r_O}=\frac{r_S\sin\theta_S \sin\phi_S}{r_O-r_S\sin\theta_S\cos\phi_S}$}
\end{equation}%

\begin{equation}
\fbox{$\displaystyle y_S:=\frac{\beta_S}{r_O}=\frac{-r_S \cos\theta_S}{r_O-r_S\sin\theta_S\cos\phi_S} $}
\end{equation}%
or equivalently:%
\begin{equation}
\frac{\alpha _{S}}{\beta _{S}}=-\tan \theta _{S}\sin \phi _{S}
\end{equation}

\subsection{Solution of the lens equation and the computation of $\protect%
\theta _{S},\protect\phi _{S}$, $\protect\alpha _{i},\protect\beta _{i}.$}

We now describe how we solve the lens equations (\ref{Lenseksiswsi}) using
the properties of the Weierstra\ss\ Jacobi modular form $\wp (z)$ equation (%
\ref{WeierstrassKarl}) and the computation of the radial and angular
integrals in terms of Appell-Lauricella hypergeometric functions equations (%
\ref{R1ROne}),(\ref{Lauricella4}),(\ref{GvniaKleisti1}),(\ref%
{KleistiGwniaki2}) respectively.

For a choice of initial conditions $a,\Phi ,\mathcal{Q}$ we determine values
for the observer image plane coordinates $\alpha _{i},\beta _{i}$, see
equation (\ref{Isimerinieikona}). Subsequently we determine the value of $%
z_{S}$ and therefore of $\theta _{S}$ that satisfies the equation \footnote{%
Which is the closed form solution of the radial and angular integrals of the
first of lens equations in eqn (\ref{Lenseksiswsi}).} :%
\begin{eqnarray}
2\int_{\alpha }^{\infty }\frac{1}{\sqrt{R}}\mathrm{d}r
&=&A_{1}(x_{i},y_{i},x_{S},y_{S},m)\Leftrightarrow \frac{2}{\sqrt{(\alpha
-\gamma )(\alpha -\delta )}}\frac{\Gamma (1/2)}{\Gamma (3/2)}F_{1}\left( 
\frac{1}{2},\frac{1}{2},\frac{1}{2},\frac{3}{2},\frac{\beta -\gamma }{\alpha
-\gamma },\frac{\delta -\beta }{\delta -\alpha }\right) =  \notag \\
&&2(m-1)\frac{1}{2|a|}\sqrt{\frac{z_{m}}{z_{m}(z_{m}-z_{3})}}\pi F\left( 
\frac{1}{2},\frac{1}{2},1,\frac{z_{m}}{z_{m}-z_{3}}\right) +  \notag \\
&&\frac{1}{2|a|}\frac{\sqrt[2]{(z_{m}-z_{S})}}{\sqrt[2]{z_{m}(z_{m}-z_{3})}}%
F_{1}\left( \frac{1}{2},\frac{1}{2},\frac{1}{2},\frac{3}{2},\frac{z_{m}-z_{S}%
}{z_{m}},\frac{z_{m}-z_{S}}{z_{m}-z_{3}}\right) \frac{\Gamma (\frac{1}{2}%
)\Gamma (1)}{\Gamma (3/2)}  \notag \\
&&+[1-\mathrm{sign(}\theta _{S}\circ \theta _{mS})]\frac{1}{|a|}\frac{\sqrt[2%
]{\frac{z_{S}(z_{m}-z_{3})}{z_{m}(z_{S}-z_{3})}}}{\sqrt[2]{z_{m}-z_{3}}}%
\times  \notag \\
&&F_{1}\left( \frac{1}{2},\frac{1}{2},\frac{1}{2},\frac{3}{2},\frac{z_{m}}{%
z_{m}-z_{3}}\frac{z_{S}(z_{m}-z_{3})}{z_{m}(z_{S}-z_{3})},\frac{%
z_{S}(z_{m}-z_{3})}{z_{m}(z_{S}-z_{3})}\right) +  \notag \\
&&\frac{1}{2|a|}\frac{\sqrt[2]{1}}{\sqrt[2]{(z_{m}-z_{3})}}F\left( \frac{1}{2%
},\frac{1}{2},1,\frac{z_{m}}{z_{m}-z_{3}}\right) \pi  \notag \\
&&  \label{Constraint1Equat} \\
&&  \notag
\end{eqnarray}%
using our exact solution for $z_{S\text{ }}$ in terms of the Weierstra\ss\ %
elliptic function equation (\ref{WeierstrassKarl}). For this we need to know
at which regions of the fundamental period parallelogram \ the Weierstra\ss\ %
function takes real and negative values. Indeed, the function of Weierstra%
\ss\ takes the required values at the points: $x=\frac{\omega }{l}+\omega
^{\prime },l\in 
\mathbb{R}
$ of the fundamental region ($\wp (\frac{\omega }{l}+\omega ^{\prime
};g_{2},g_{3})\in 
\mathbb{R}
^{-}$)$.$ Thus as the parameter $l$ varies we determine the value of $z_{S}$
that satisfies equations (\ref{WeierstrassKarl}), (\ref{Constraint1Equat}).
The quantities $\omega ,\omega ^{\prime }$ denote the Weierstra\ss\ %
half-periods. In the case under investigation $\omega $ is a real
half-period while $\omega ^{\prime }$ is pure imaginary. For positive
discriminant $\Delta _{c}=g_{2}^{3}-27g_{3}^{2},$ all three roots $%
e_{1},e_{2},e_{3}$ of $4z^{3}-g_{2}z-g_{3}$ are real and if the $e_{i}$ are
ordered so that $e_{1}>e_{2}>e_{3}$ we can choose the half-periods as%
\begin{equation}
\omega =\int_{e_{1}}^{\infty }\frac{\mathrm{d}t}{\sqrt{4t^{3}-g_{2}t-g_{3}}}%
,\quad \omega ^{\prime }=\mathrm{i}\int_{-\infty }^{e_{3}}\frac{\mathrm{d}t}{%
\sqrt{-4t^{3}+g_{2}t+g_{3}}}
\end{equation}%
The period ratio $\tau $ is defined by $\tau =\omega ^{\prime }/\omega .$ An
alternative expression for the real half-period $\omega $ of the Weierstra%
\ss\ elliptic function is given by the hypergeometric function of Gau\ss \footnote{The three roots are given in terms of the first integrals of motion by the expressions: $e_1=\frac{1}{24}(-a^2+{\cal Q}+\Phi^2+3\sqrt{4a^2 {\cal Q}+
(-a^2+{\cal Q}+\Phi^2)^2})$, $
e_2=\frac{1}{12}(a^2-{\cal Q}-\Phi^2)$, 
$e_3=\frac{1}{24}(-a^2+{\cal Q}+\Phi^2
-3\sqrt{4a^2 {\cal Q}+(-a^2+{\cal Q}+\Phi^2)^2})$}:
\begin{equation}
\omega =\frac{1}{\sqrt{e_{1-}e_{3}}}\frac{\pi }{2}F\left( \frac{1}{2},\frac{1%
}{2},1,\frac{e_{2}-e_{3}}{e_{1}-e_{3}}\right)
\end{equation}

Having determined $\theta _{S}$ by the procedure we just described we
determine the azimuthial position of the source $\phi _{S}$ by the second
equation of (\ref{Lenseksiswsi}):%
\begin{equation}
-\phi _{S}=R_{2}(x_{i},y_{i})+A_{2}(x_{i},y_{i},x_{S},y_{S},m)
\label{AzimuPigis}
\end{equation}%
Let us give an example. For the choice $\mathcal{Q}=24.64563\frac{G^{2}M^{2}%
}{c^{4}},\Phi =-2.719110\frac{GM}{c^{2}},a=0.6\frac{GM}{c^{2}}$ we determine 
$z_{S}=0.3161007914992452,$ $m=3$ and $\Delta \phi =-11.086,\phi _{S}
=95.1794^{\circ}.$ Keeping fixed the value of the Kerr parameter we solved the lens
equations for different values of Carter's constant $\mathcal{Q}$ and impact
factor $\Phi .$ We exhibit our results in table \ref{M2A06}\footnote{%
Assuming that the galactic centre region, SgrA* , is a Kerr black hole with
mass: $M_{\mathrm{BH}}=10^{6}%
M_{\odot}%
$and a distance from the observer to the galactic centre: $r_{O}=8$Kpc, the
second solution in Table \ref{M2A06}will require an angular resolution of $%
19.3102\mu \mathrm{arcs}$ . This is in the range of experimental accuracy
for both the TMT and GRAVITY experiments.}.

\begin{table}[tbp] \centering%
\begin{tabular}{|l|l|l|}
\hline
& $a=0.6,\mathcal{Q}=24.64563,$ $\Phi =-2.719110$ & $a=0.6,\mathcal{Q}%
=0.128, $ $\Phi =3.839$ \\ \hline
$\alpha _{i}$ ($\frac{GM}{c^{2}}$) & $2.719110$ & $-3.839$ \\ 
$\beta _{i}$ $(\frac{GM}{c^{2}})$ & $-4.9644365239$ & $0.357770876399$ \\ 
$x_{i}\left( \frac{2}{r_{O}}\frac{GM}{c^{2}}\right) $ & $1.359555$ & $%
-1.9195 $ \\ \hline
$y_{i}(\frac{2}{r_{O}}\frac{GM}{c^{2}}\mathrm{)}$ & $-2.48221826$ & $%
0.178885 $ \\ \hline
$m$ & $3$ & $3$ \\ \hline
&  &  \\ \hline
$z_{S}$ & $0.3161007914992452$ & $0.0026145818604$ \\ \hline
$\theta _{S}$ & $55.79^{\circ}$ & $87.069^{\circ}$ \\ \hline
$\Delta \phi (\rm{rad})$ & $-11.086$ & $7.09441$ \\ \hline
$\phi _{S}$ & $95.1794^{\circ}$ & $133.52^{\circ}$ \\ \hline
$\omega $ & $0.5545341990201503500$ & $0.824718843878947$ \\ \hline
$\omega ^{\prime }$ & $\text{ }1.3278669366032567973\mathrm{i}$ & $%
2.9400828459149726\mathrm{i}$ \\ \hline
\end{tabular}%
\caption{Solution of the lens equations in Kerr geometry and the predictions
for the source and image positions for an observer at $\theta_O=\pi/2,
\phi_O=0$. The number of turning points in the polar variable is three. The
values for the Kerr parameter and the impact factor $\Phi$ are in units of
$\frac{GM}{c^2}$  while those of Carter's constant ${\cal Q}$ are in units of
$\frac{G^2M^2}{c^4}$. }\label{M2A06}%
\end{table}%

Let us at this point, present a solution with a higher value for
the Kerr parameter in Table \ref{09939ObsIsime}.

\begin{figure}
\begin{center}
\includegraphics{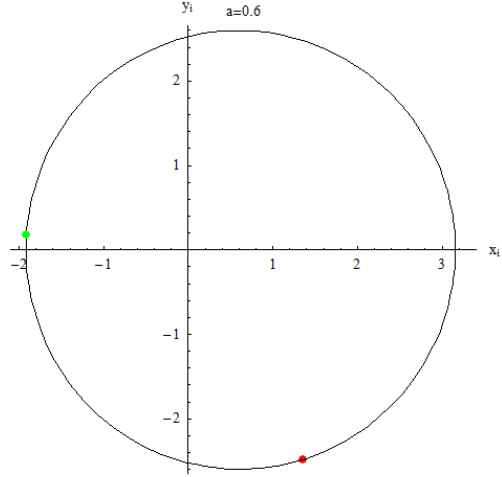}
\caption{The two images of Table \ref{M2A06} on the observer's image plane. The value of the Kerr parameter is $a=0.6\frac{GM}{c^2}$, while the observer is located at $\theta_O=\pi/2$. With red is the image solution, first column of Table \ref{M2A06} and with green the image solution, second column 
of table \ref{M2A06}.
\label{Shadow06EquOb}}
\end{center}
\end{figure}

\begin{table}[tbp] \centering%
\begin{tabular}{|l|l|}
\hline
& $a=0.9939,$ $\mathcal{Q}=27.0220588123,$ $\Phi =-2.29885534$ \\ \hline
$\alpha _{i}$ ($\frac{GM}{c^{2}}$) & $2.29885534$ \\ 
$\beta _{i}$ $(\frac{GM}{c^{2}})$ & $5.198274599547431$ \\ 
$x_{i}\left( \frac{2}{r_{O}}\frac{GM}{c^{2}}\right) $ & $1.14942767$ \\ 
\hline
$y_{i}(\frac{2}{r_{O}}\frac{GM}{c^{2}}\mathrm{)}$ & $2.5991372997737154$ \\ 
\hline
$m$ & $3$ \\ \hline
&  \\ \hline
$z_{S}$ & $0.01378435185109$ \\ \hline
$\theta _{S}$ & $83.2575^{\circ}$ \\ \hline
$\Delta \phi (\rm{rad})$ & $-11.243$ \\ \hline
$\phi _{S}$ & $104.177^{\circ}$ \\ \hline
$\omega $ & $0.5505433970950226$ \\ \hline
$\omega ^{\prime }$ & $1.1288708298860726$ $\mathrm{i}$ \\ \hline
\end{tabular}%
\caption{Solution of the lens equations in Kerr geometry and the predictions
for the source and image positions for an observer at $\theta_O=\pi/2,
\phi_O=0$ for a high value for the spin of the black hole. The number of turning points in the polar variable is three. The
values for the Kerr parameter and the impact factor $\Phi$ are in units of
$\frac{GM}{c^2}$ while those of Carter's constant ${\cal Q}$ are in units of
$\frac{G^2M^2}{c^4}$ . }\label{09939ObsIsime}%
\end{table}%

The positions of the images of Tables \ref{M2A06} and \ref{09939ObsIsime} on
the observer's image plane are displayed in fig.\ref{Shadow06EquOb} and 
fig.\ref{HighspinIsim} respectively. In the same figures the boundary of the
shadow of the spinning black hole is also displayed. 

\begin{figure}
\begin{center}
\includegraphics{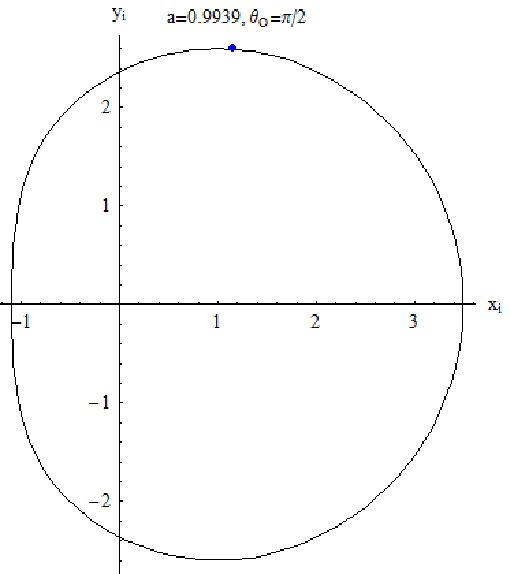}
\caption{The lens solution of Table \ref{09939ObsIsime} as it will be detected on the observer image plane by an equatorial observer. The boundary of the shadow of the black hole is also exhibited.
\label{HighspinIsim}}
\end{center}
\end{figure}

\subsection{Closed form calculation for the magnifications.}

We outline in this subsection, the closed-form calculation, of the resulting
magnification factors. It turns out that the derivatives involved in the
expression for the magnification are elegantly computed using the beautiful
property of the hypergeometric functions: namely, that the derivatives of
the hypergeometric functions of Appell-Lauricella are again hypergeometric
functions of the same type with a different set of parameters. It is a
powerful property of our formalism which we exploit to the full in what
follows.

\begin{eqnarray}
\frac{\partial (\ref{gwniakimetasximaolokliro})}{\partial x_{S}} &=&\frac{%
\partial (\ref{gwniakimetasximaolokliro})}{\partial z_{S}}\frac{\partial
z_{S}}{\partial x_{S}},  \notag \\
\frac{\partial (\ref{gwniakimetasximaolokliro})}{\partial z_{S}} &=&\frac{%
\Phi }{2|a|}\frac{1}{z_{m}}\frac{1}{(1-z_{m})}\frac{1}{\sqrt{z_{m}-z_{3}}}%
\left( \frac{z_{m}-z_{S}}{z_{m}}\right) ^{-1/2}\times  \notag \\
&&F_{D}\left( \frac{1}{2},1,\frac{1}{2},\frac{1}{2},\frac{3}{2},\frac{%
z_{S}-z_{m}}{1-z_{m}},\frac{z_{m}-z_{S}}{z_{m}},\frac{z_{m}-z_{S}}{%
z_{m}-z_{3}}\right) +  \notag \\
&&\left( \frac{-\Phi }{2|a|}\sqrt[2]{\frac{(z_{m}-z_{S})}{z_{m}}}\frac{1}{%
\sqrt[2]{z_{m}-z_{3}}}\frac{2}{(1-z_{m})}\right) \times 
\Biggl\{
\notag \\
&&F_{D}\left( \frac{3}{2},2,\frac{1}{2},\frac{1}{2},\frac{5}{2},\frac{%
z_{S}-z_{m}}{1-z_{m}},\frac{z_{m}-z_{S}}{z_{m}},\frac{z_{m}-z_{S}}{%
z_{m}-z_{3}}\right) \frac{1}{1-z_{m}}+  \notag \\
&&F_{D}\left( \frac{3}{2},1,\frac{3}{2},\frac{1}{2},\frac{5}{2},\frac{%
z_{S}-z_{m}}{1-z_{m}},\frac{z_{m}-z_{S}}{z_{m}},\frac{z_{m}-z_{S}}{%
z_{m}-z_{3}}\right) \frac{-1}{z_{m}}+  \notag \\
&&F_{D}\left( \frac{3}{2},1,\frac{1}{2},\frac{3}{2},\frac{5}{2},\frac{%
z_{S}-z_{m}}{1-z_{m}},\frac{z_{m}-z_{S}}{z_{m}},\frac{z_{m}-z_{S}}{%
z_{m}-z_{3}}\right) \frac{-1}{z_{m}-z_{3}}%
\Biggr\}%
+  \notag \\
&&\left( 1-\mathrm{sign(}\theta _{S}\circ \theta _{ms}\right) (-1)%
\Biggl[%
\Bigl[%
\frac{1}{z_{m}}\frac{z_{S}-z_{m}}{1-z_{S}}\frac{1}{\sqrt{%
z_{S}(z_{S}-z_{m})(z_{3}-z_{S})}}+  \notag \\
&&\frac{z_{S}}{z_{m}}\frac{%
z_{3}(z_{m}-3z_{S}z_{m}+2z_{S}^{2})-z_{S}(z_{m}(2-4z_{S})+z_{S}(-1+3z_{S}))}{%
2(1-z_{S})^{2}z_{S}(z_{3}-z_{S})\sqrt{z_{S}(z_{S}-z_{m})(z_{3}-z_{S})}}%
\Bigr]%
\times  \notag \\
&&F_{D}\left( 1,1,-\frac{1}{2},\frac{1}{2},\frac{3}{2},\frac{z_{S}(1-z_{m})}{%
z_{m}(1-z_{S})},\frac{z_{S}}{z_{m}},\frac{z_{S}(z_{m}-z_{3})}{%
z_{m}(z_{S}-z_{3)}}\right) +  \notag \\
&&\frac{z_{S}}{z_{m}}\frac{z_{S}-z_{m}}{(1-z_{S})}\frac{1}{\sqrt{%
z_{S}(z_{S}-z_{m})(z_{3}-z_{S})}}%
\Bigl\{
\notag \\
&&F_{D}\left( 2,2,-\frac{1}{2},\frac{1}{2},\frac{5}{2},\frac{z_{S}(1-z_{m})}{%
z_{m}(1-z_{S})},\frac{z_{S}}{z_{m}},\frac{z_{S}(z_{m}-z_{3})}{%
z_{m}(z_{S}-z_{3)}}\right) \frac{1-z_{m}}{z_{m}(1-z_{S})^{2}}+  \notag \\
&&F_{D}\left( 2,1,\frac{1}{2},\frac{1}{2},\frac{5}{2},\frac{z_{S}(1-z_{m})}{%
z_{m}(1-z_{S})},\frac{z_{S}}{z_{m}},\frac{z_{S}(z_{m}-z_{3})}{%
z_{m}(z_{S}-z_{3)}}\right) \frac{1}{z_{m}}+  \notag \\
&&F_{D}\left( 2,1,\frac{-1}{2},\frac{3}{2},\frac{5}{2},\frac{z_{S}(1-z_{m})}{%
z_{m}(1-z_{S})},\frac{z_{S}}{z_{m}},\frac{z_{S}(z_{m}-z_{3})}{%
z_{m}(z_{S}-z_{3)}}\right) \left( \frac{-z_{3}(z_{m}-z_{3})}{%
z_{m}(z_{S}-z_{3})^{2}}\right) 
\Bigr\}%
\Biggr]
\notag \\
&&  \label{PA2Pxs}
\end{eqnarray}%
Thus,%
\begin{equation}
\frac{\partial (\ref{gwniakimetasximaolokliro})}{\partial x_{S}}=(\ref%
{PA2Pxs})\times \left( -2\cos \theta _{S}\sin \theta _{S}\frac{\frac{%
r_{S}^{2}\sin \theta _{S}\cos \theta _{S}\sin \phi _{S}}{(r_{O}-r_{S}\sin
\theta _{S}\cos \phi _{S})^{2}}}{J_{1}}\right)
\end{equation}%
Now we calculate the term:$\frac{\partial (\ref{SimplerMarker})}{\partial
z_{S}}.$ Indeed, \ calculating the derivatives w.r.t. $z_{S}$ we derive the
expression:%
\begin{eqnarray}
\frac{\partial (\ref{SimplerMarker})}{\partial z_{S}} &=&\frac{1}{2|a|}\frac{%
\Gamma (1)\Gamma (1/2)}{\Gamma (3/2)}\left( -\frac{1}{2\sqrt{%
z_{m}(z_{m}-z_{3})}\sqrt{z_{m}-z_{S}}}\right) F_{1}\left( \frac{1}{2},\frac{1%
}{2},\frac{1}{2},\frac{3}{2},\frac{z_{m}-z_{S}}{z_{m}},\frac{z_{m}-z_{S}}{%
z_{m}-z_{3}}\right) +  \notag \\
&&\frac{1}{2|a|}\frac{\Gamma (1)\Gamma (1/2)}{\Gamma (3/2)}\sqrt{\frac{%
(z_{m}-z_{S})}{z_{m}(z_{m}-z_{3})}}\times 
\Biggl[%
F_{1}\left( \frac{3}{2},\frac{3}{2},\frac{1}{2},\frac{5}{2},\frac{z_{m}-z_{S}%
}{z_{m}},\frac{z_{m}-z_{S}}{z_{m}-z_{3}}\right) \left( \frac{-1}{z_{m}}%
\right) +  \notag \\
&&F_{1}\left( \frac{3}{2},\frac{1}{2},\frac{3}{2},\frac{5}{2},\frac{%
z_{m}-z_{S}}{z_{m}},\frac{z_{m}-z_{S}}{z_{m}-z_{3}}\right) \left( \frac{-1}{%
z_{m}-z_{3}}\right) 
\Biggr]%
+  \notag \\
&&\left[ 1-\mathrm{sign(}\theta _{S}\circ \theta _{ms})\right] 
\Biggl[%
\frac{1}{2|a|}\left( \frac{z_{S}(z_{m}-z_{3})}{z_{m}(z_{S}-z_{3})}\right) ^{-%
\frac{1}{2}}\left\{ \frac{(-z_{3})\sqrt{z_{m}-z_{3}}}{z_{m}(z_{S}-z_{3})^{2}}%
\right\} \times  \notag \\
&&F_{1}\left( \frac{1}{2},\frac{1}{2},\frac{1}{2},\frac{3}{2},\frac{z_{m}}{%
z_{m}-z_{3}}\frac{z_{S}(z_{m}-z_{3})}{z_{m}(z_{S}-z_{3})},\frac{%
z_{S}(z_{m}-z_{3})}{z_{m}(z_{S}-z_{3})}\right) +  \notag \\
&&\frac{1}{|a|}\frac{\sqrt{\frac{z_{S}(z_{m}-z_{3})}{z_{m}(z_{S}-z_{3})}}}{%
\sqrt{z_{m}-z_{3}}}\times 
\Bigl[%
F_{1}\left( \frac{3}{2},\frac{3}{2},\frac{1}{2},\frac{5}{2},\frac{z_{m}}{%
z_{m}-z_{3}}\frac{z_{S}(z_{m}-z_{3})}{z_{m}(z_{S}-z_{3})},\frac{%
z_{S}(z_{m}-z_{3})}{z_{m}(z_{S}-z_{3})}\right) \left( \frac{-z_{3}}{%
(z_{S}-z_{3})^{2}}\right) +  \notag \\
&&F_{1}\left( \frac{3}{2},\frac{1}{2},\frac{3}{2},\frac{5}{2},\frac{z_{m}}{%
z_{m}-z_{3}}\frac{z_{S}(z_{m}-z_{3})}{z_{m}(z_{S}-z_{3})},\frac{%
z_{S}(z_{m}-z_{3})}{z_{m}(z_{S}-z_{3})}\right) \left( \frac{%
(-z_{3})(z_{m}-z_{3})}{z_{m}(z_{S}-z_{3})^{2}}\right) 
\Bigr]%
\Biggr]
\notag \\
&&  \notag \\
&&  \label{MagA1}
\end{eqnarray}%
Now:

\begin{equation}
\fbox{$\displaystyle \alpha_1 =\frac{\partial A_1}{\partial x_S}=(\ref{MagA1})\times\frac{\partial z_S}{\partial x_S}  
= (\ref{MagA1})\times\Biggl( -2\cos\theta_S\sin\theta_S\times
\frac{\frac{r_S^2 \sin\theta_S \cos\theta_S \sin\phi_S}{(r_O-r_S \sin\theta_S\cos\phi_S)^2}}{J_1} \Biggr)
\label{alphaena}$}
\end{equation}%
and

\begin{equation}
\fbox{$\displaystyle \alpha_2=\frac{\partial A_1}{\partial y_S}=(\ref{MagA1})\times\frac{\partial z_S}{\partial y_S}  
= (\ref{MagA1})\times\Biggl( -2\cos\theta_S\sin\theta_S\times
\frac{\frac{-[r_O r_S \sin\theta_S \cos\phi_S-r_S^2 \sin^2\theta_S]}{(r_O-r_S \sin\theta_S\cos\phi_S)^2}}{J_1} \Biggr)$}
\label{alphadyo}
\end{equation}%

\bigskip While for the $\alpha _{3},\alpha _{4}$ terms which contrbute to
the expression for the magnification. equation (\ref{MegenthysiVaritiki}),
we derive the expressions:%
\begin{equation}
\fbox{$\displaystyle\alpha_3=-\frac{\partial \phi_S}{\partial x_S}-\frac{\partial A_2}{\partial x_S} 
=-\Biggl(-\frac{\frac{(r_Or_S\sin\theta_S-r_S^2\cos\phi_S)}{(r_O-r_S\sin\theta_S\cos\phi_S)^2}}{J_1}\Biggr) 
-(\ref{PA2Pxs})\times \Biggl(\frac{\frac{r_S^2 \sin\theta_S \cos\theta_S \sin\phi_S}{(r_O-r_S \sin\theta_S\cos\phi_S)^2}}{J_1}\Biggr)$}
\label{alphatria}
\end{equation}%

\begin{equation}
\fbox{$\displaystyle \alpha_4=-\frac{\partial \phi_S}{\partial y_S}-\frac{\partial A_2}{\partial y_S}=
-\frac{\frac{r_Or_S\cos\theta_S\sin\phi_S}{(r_O-r_S\sin\theta_S\cos\phi_S)^2}}{J_1}-(\ref{PA2Pxs})\times 
\frac{\frac{-[r_O r_S \sin\theta_S \cos\phi_S-r_S^2 \sin^2\theta_S]}{(r_O-r_S \sin\theta_S\cos\phi_S)^2}}{J_1} \Biggr)$}
\label{alphatessera}
\end{equation}%

\bigskip

\bigskip

\bigskip

where $J_{1}$ denotes the Jacobian:%
\begin{equation}
J_{1}=\frac{\partial (x_{S},y_{S})}{\partial (\theta _{S},\phi _{S})}
\end{equation}%
and%
\begin{eqnarray}
\frac{\partial \theta _{S}}{\partial x_{S}} &=&\frac{(r_{S}^{2}\sin \theta
_{S}\cos \theta _{S}\sin \phi _{S})/((r_{O}-r_{S}\sin \theta _{S}\cos \phi
_{S})^{2})}{J_{1}}  \notag \\
\frac{\partial \theta _{S}}{\partial y_{S}} &=&\frac{-[r_{O}r_{S}\sin \theta
_{S}\cos \phi _{S}-r_{S}^{2}\sin ^{2}\theta _{S}]/((r_{O}-r_{S}\sin \theta
_{S}\cos \phi _{S})^{2})}{J_{1}}  \notag \\
\frac{\partial \phi _{S}}{\partial x_{S}} &=&-\frac{(r_{O}r_{S}\sin \theta
_{S}-r_{S}^{2}\cos \phi _{S})/((r_{O}-r_{S}\sin \theta _{S}\cos \phi
_{S})^{2})}{J_{1}}  \notag \\
\frac{\partial \phi _{S}}{\partial y_{S}} &=&\frac{r_{O}r_{S}\cos \theta
_{S}\sin \phi _{S}/((r_{O}-r_{S}\sin \theta _{S}\cos \phi _{S})^{2})}{J_{1}}
\end{eqnarray}

In producing the results exhibited in eqns (\ref{PA2Pxs}),(\ref{MagA1}) in
our calculations for the magnification factors we made use of the important
identity of Appell's hypergeometric function $F_{1}:$

$%
\fbox
{\begin{Beqnarray} \frac{\partial^{m+n} F_1 (\alpha,\beta,\beta^{\prime},\gamma,x,y)}
{\partial x^m \partial x^n}&=&
\frac{(\alpha,m+n)(\beta,m)(\beta^{\prime},n)}{(\gamma,m+n)}\times \\ \nonumber
& & F_1 (\alpha+m+n,\beta+m,\beta^{\prime}+n,\gamma+m+n,x,y) \end{Beqnarray}}%
$ and its corresponding generalization for the fourth hypergeometric
function of Lauricella. Similar calculations that we do not exhibit in this
written account of our talk, lead to the derivation of the coefficients $%
\beta _{1},\beta _{2},\beta _{3},\beta _{4}$ in terms of the generalized
hypergeometric functions of Appell-Lauricella.

A phenomenological analysis of our exact solutions for the magnifications in
Kerr spacetime will be a subject of a separate publication \cite{KraniotisGV}%
.

\subsection{Source and image positions for an observer located at $\protect%
\theta _{O}=\frac{\protect\pi }{3}\label{60moiresparatiritis}.$}

In this case, the coordinates on the observers image plane are related to
the first integrals of motions as follows:

\begin{equation}
\Phi =-\alpha _{i}\frac{\sqrt{3}}{2},\quad Q=\beta _{i}^{2}+\left( \frac{%
4\Phi ^{2}}{3}-a^{2}\right) \frac{1}{4}.
\end{equation}%
Furthermore, our solution for the first lens equation (\ref{Constraint1one})
takes the form:
 
\begin{eqnarray}
2\int_{\alpha }^{\infty }\frac{1}{\sqrt{R}}\mathrm{d}r
&=&A_{1}(x_{i},y_{i},x_{S},y_{S},m)\Leftrightarrow \frac{2}{\sqrt{(\alpha
-\gamma )(\alpha -\delta )}}\frac{\Gamma (1/2)}{\Gamma (3/2)}F_{1}\left( 
\frac{1}{2},\frac{1}{2},\frac{1}{2},\frac{3}{2},\frac{\beta -\gamma }{\alpha
-\gamma },\frac{\delta -\beta }{\delta -\alpha }\right) =  \notag \\
&&2(m-1)\frac{1}{2|a|}\sqrt{\frac{z_{m}}{z_{m}(z_{m}-z_{3})}}\pi F\left( 
\frac{1}{2},\frac{1}{2},1,\frac{z_{m}}{z_{m}-z_{3}}\right) +  \notag \\
&&\frac{1}{2|a|}\frac{\sqrt[2]{(z_{m}-z_{S})}}{\sqrt[2]{z_{m}(z_{m}-z_{3})}}%
F_{1}\left( \frac{1}{2},\frac{1}{2},\frac{1}{2},\frac{3}{2},\frac{z_{m}-z_{S}%
}{z_{m}},\frac{z_{m}-z_{S}}{z_{m}-z_{3}}\right) \frac{\Gamma (\frac{1}{2}%
)\Gamma (1)}{\Gamma (3/2)}  \notag \\
&&+[1-\mathrm{sign(}\theta _{S}\circ \theta _{mS})]\frac{1}{|a|}\frac{\sqrt[2%
]{\frac{z_{S}(z_{m}-z_{3})}{z_{m}(z_{S}-z_{3})}}}{\sqrt[2]{z_{m}-z_{3}}}%
\times  \notag \\
&&F_{1}\left( \frac{1}{2},\frac{1}{2},\frac{1}{2},\frac{3}{2},\frac{z_{m}}{%
z_{m}-z_{3}}\frac{z_{S}(z_{m}-z_{3})}{z_{m}(z_{S}-z_{3})},\frac{%
z_{S}(z_{m}-z_{3})}{z_{m}(z_{S}-z_{3})}\right) +  \notag \\
&&\frac{1}{2|a|}\frac{\sqrt[2]{(z_{m}-\frac{1}{4})}}{\sqrt[2]{%
z_{m}(z_{m}-z_{3})}}F_{1}\left( \frac{1}{2},\frac{1}{2},\frac{1}{2},\frac{3}{%
2},\frac{z_{m}-\frac{1}{4}}{z_{m}},\frac{z_{m}-\frac{1}{4}}{z_{m}-z_{3}}%
\right) \frac{\Gamma (\frac{1}{2})\Gamma (1)}{\Gamma (3/2)}  \notag \\
&&+[1-\mathrm{sign(}\frac{\pi }{3}\circ \theta _{mO})]\frac{1}{|a|}\frac{%
\sqrt[2]{\frac{(1/4)(z_{m}-z_{3})}{z_{m}((1/4)-z_{3})}}}{\sqrt[2]{z_{m}-z_{3}%
}}\times  \notag \\
&&F_{1}\left( \frac{1}{2},\frac{1}{2},\frac{1}{2},\frac{3}{2},\frac{z_{m}}{%
z_{m}-z_{3}}\frac{(1/4)(z_{m}-z_{3})}{z_{m}(1/4-z_{3})},\frac{%
(1/4)(z_{m}-z_{3})}{z_{m}((1/4)-z_{3})}\right) .  \label{60Constraint1}
\end{eqnarray}%
Let us give now examples of solutions for the images and source positions of
equations (\ref{60Constraint1}),(\ref{WeierstrassKarl}).

For the initial conditions $\mathcal{Q=}25.64563,a=0.9939,\Phi =-3.11$ we
calculate the parameter $z_{S}$ of the source latitude position to be: $%
z_{S}=0.09097820848$ ($\theta _{S}=72.4447^{\circ}$) for $m=3.$ The position at the fundamental period parallelogram that
provides the above value of $z_{S}$ as a solution of (\ref{60Constraint1}),(%
\ref{WeierstrassKarl}) for three turning points in the polar variable is
located at:$\frac{\omega }{1.2995480690017123}+\omega ^{\prime },$ where the
fundamental half-periods of the Weierstra\ss\ elliptic function $\wp
(x,g_{2},g_{3})$ were calculated to be:%
\begin{equation}
\omega =0.52792338858688228,\quad \omega ^{\prime }=1.119903617249492\text{ }%
\mathrm{i}
\end{equation}%
Also the azimuthial position of the source was calculated to be using (\ref%
{Constraint2}) and the calculated value of $z_{S}$: $\phi _{S}=2.67231589%
\rm{rad}=153.112^{\circ}
=551205^{\prime \prime}\footnote{$\Delta \phi
=R_{2}(x_{i},y_{i})+A_{2}(x_{i},y_{i},x_{s},y_{s},m) $ was calculated to be: 
$\Delta \phi =-12.0971\rm{rad}$ so that the photons perform more than one
loop and a half around the black hole$.$}.$

\bigskip

\begin{table}[tbp] \centering%
\begin{tabular}{|l|l|l|}
\hline
& $a=0.9939,\mathcal{Q}=25.64563,$ $\Phi =-3.11$ & $a=0.52,\mathcal{Q}%
=23.64563,$ $\Phi =-2.85$ \\ \hline
$\alpha _{i}$ ($\frac{GM}{c^{2}}$) & $3.591118674$ & $3.29089$ \\ 
$\beta _{i}$ $(\frac{GM}{c^{2}})$ & $-4.7611506980$ & $-4.58320084657$ \\ 
$x_{i}\left( \frac{2}{r_{O}}\frac{GM}{c^{2}}\right) $ & $1.79556$ & $1.64545$
\\ \hline
$y_{i}(\frac{2}{r_{O}}\frac{GM}{c^{2}}\mathrm{)}$ & $-2.38058$ & $-2.2916004$
\\ \hline
$m$ & $3$ & $3$ \\ \hline
&  &  \\ \hline
$z_{S}$ & $0.09097820848$ & $0.5980171072414$ \\ \hline
$\theta _{S}$ & $72.4447^{\circ}$ & $39.3474^{\circ}$ \\ \hline
$\Delta \phi (\rm{rad})$ & $-12.0971$ & $-11.8577$ \\ \hline
$\phi _{S}$ & $153.112^{\circ}$ & $139.395^{\circ}$ \\ \hline
$\omega $ & $0.52792338858688228$ & $0.5571026427501503$ \\ \hline
$\omega ^{\prime }$ & $1.119903617249492\text{ }\mathrm{i}$ & $%
1.389041935594241\mathrm{i}$ \\ \hline
\end{tabular}%
\caption{Solution of the lens equations in Kerr geometry and the predictions
for the source and image positions for an observer at $\theta_O=\pi/3,
\phi_O=0$. The number of turning points in the polar variable is three. The
values for the Kerr parameter and the impact factor are in units of
$\frac{GM}{c^2}$. }\label{M2A099616}%
\end{table}%

\bigskip The first solution is shown on the image plane of the observer, 
Fig.\ref{60moiresHS}. We
observe that the solution lies close to the boundary of the shadow of the
black hole.

\begin{figure}
\begin{center}
\includegraphics{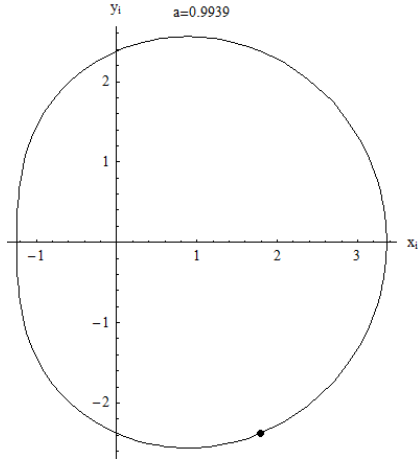}
\caption{The solution, $1^{st}$ column of Table \ref{M2A099616} as it 
will be detected on the observer image plane by an observer at $\theta_O=\pi/3,
\phi_O=0$. The boundary of the shadow of the black hole is also exhibited.
\label{60moiresHS}}
\end{center}
\end{figure}

\section{ Exact solution of the angular integrals in the
presence of the cosmological constant $\Lambda .$}

There has been a discussion in the literature as to whether or not the
cosmological constant contributes to the gravitational lensing. However, the
debate has been 
restricted
to the Schwarzschild-de Sitter spacetime \cite{Lake}, \cite{Sereno},\cite%
{Rindler}. Let us discuss now the more general case of gravitational lensing
in the Kerr-de Sitter spacetime.

The generalized solution for the angular integral (\ref{Sourceangular}) in
the presence of $\Lambda $ is given by:%
\begin{eqnarray}
\pm \int_{\theta _{S}}^{\theta _{\min /\max }} &=&\frac{\Xi ^{2}}{2|H|}\frac{%
z_{m}-z_{S}}{(1-\eta z_{m})}\frac{1}{\sqrt[2]{z_{m}(z_{m}-z_{S})(z_{m}-z_{3})%
}}\times 
\Biggl\{
\notag \\
&&\frac{\Phi }{(1-z_{m})}F_{D}\left( \frac{1}{2},1,1,\frac{1}{2},\frac{1}{2},%
\frac{3}{2},\frac{\eta (z_{S}-z_{m})}{1-\eta z_{m}},\frac{z_{S}-z_{m}}{%
1-z_{m}},\frac{z_{m}-z_{S}}{z_{m}},\frac{z_{m}-z_{S}}{z_{m}-z_{3}}\right) 
\frac{\Gamma \left( \frac{1}{2}\right) }{\Gamma \left( \frac{3}{2}\right) }+
\notag \\
&&-aF_{D}\left( \frac{1}{2},1,\frac{1}{2},\frac{1}{2},\frac{3}{2},\frac{\eta
(z_{S}-z_{m})}{1-\eta z_{m}},\frac{z_{m}-z_{S}}{z_{m}},\frac{z_{m}-z_{S}}{%
z_{m}-z_{3}}\right) \frac{\Gamma \left( \frac{1}{2}\right) }{\Gamma \left( 
\frac{3}{2}\right) }%
\Biggr\}
\notag \\
&&(1-\mathrm{sign(\theta }_{S}\circ \theta _{mS}))%
\Biggl[%
\frac{\Xi ^{2}}{\left\vert H\right\vert }\frac{z_{S}}{z_{m}}\frac{z_{m}-z_{S}%
}{1-\eta z_{S}}\frac{1}{\sqrt{z_{S}(z_{S}-z_{m})(z_{3}-z_{S})}}\times  \notag
\\
&&\left\{ -aF_{D}\left( 1,1,-\frac{1}{2},\frac{1}{2},\frac{3}{2},\lambda ,%
\frac{z_{S}}{z_{m}},\mu \right) +\frac{\Phi }{1-z_{S}}F_{D}\left( 1,1,1,-%
\frac{3}{2},\frac{1}{2},\frac{3}{2},\lambda ,\nu ,\frac{z_{S}}{z_{m}},\mu
\right) \right\} 
\Biggr]
\notag \\
&&
\end{eqnarray}%
where:

\begin{equation}
\fbox{$\displaystyle \eta:=-\frac{a^2 \Lambda}{3}, \mu=\frac{z_S}{z_m}\frac{z_m-z_3}{z_S-z_3},\lambda=
\frac{z_S}{z_m}\Bigl(\frac{1-\eta z_m}{1-\eta z_S}\Bigr), \nu=\frac{z_S}{z_m}\Bigl(\frac{1-z_m}{1-z_S}\Bigr)$}
\end{equation}%

Also the integrals $\pm \int_{\theta _{\min /\max }}^{\theta _{\max /\min
}}=2\int_{0}^{z_{m}}$ contribute the term:

\fbox
{\begin{Beqnarray}2(m-1)\times\Biggl\{&  &\frac{\Xi^2\Phi}{2|H|}\frac{z_m}{(1-\eta z_m)(1-z_m)}\frac{1}{\sqrt{z_m^2 (z_m-z_3)}}\nonumber  \\
&\times& F_D\Bigl(\frac{1}{2},1,1,\frac{1}{2},\frac{1}{2},\frac{3}{2},
\frac{\eta(-z_m)}{1-\eta z_m},\frac{-z_m}{1-z_m},1,\frac{z_m}{z_m-z_3}\Bigr)2 \nonumber \\
&+ &\frac{-\Xi^2 a}{2|H|}\frac{z_m}{\sqrt{z_m^2(z_m-z_3)}}\frac{1}{1-\eta z_m}\nonumber \\
&\times& F_D\Bigl(\frac{1}{2},1,\frac{1}{2},\frac{1}{2},\frac{3}{2},\frac{\eta(-z_m)}{1-\eta z_m},1,\frac{z_m}{z_m-z_3}\Bigr)2\Biggr\}
\end{Beqnarray}
\label{LambdaUniverse}}%

Notice that
for $\Lambda =0$ this reduces to equation (\ref{SovaroTP}).

\subsection{Closed-form solution for the radial integrals in the presence of
the cosmological constant $\Lambda .$}

Assume  $\Lambda >0.$ We need to calculate radial integrals of the form:%
\begin{equation}
\int \frac{a\Xi ^{2}}{\Delta _{r}}((r^{2}+a^{2})-a\Phi ))\frac{\mathrm{d}r}{%
\sqrt[2]{R}}
\end{equation}%
We use the technique of partial fractions from integral calculus:%
\begin{equation}
\frac{a\Xi ^{2}}{\Delta _{r}}((r^{2}+a^{2})-a\Phi ))=\frac{A^{1}}{%
r-r_{\Lambda }^{+}}+\frac{A^{2}}{r-r_{\Lambda }^{-}}+\frac{A^{3}}{r-r_{+}}+%
\frac{A^{4}}{r-r_{-}}
\end{equation}%
where $r_{\Lambda }^{+},r_{\Lambda }^{-},r_{+},r_{-}$ are the four real
roots of $\Delta _{r}.$

For instance, for 
\fbox{$r_O,r_S<r_{\Lambda}^+ $}
one of the integrals we need to calculate is:%
\begin{equation}
\frac{1}{\sqrt{\frac{1}{3}(\mathcal{Q}\Lambda +3\Xi ^{2}(1+\frac{\Lambda }{3}%
(a-\Phi )^{2})}}\int_{\alpha }^{r_{\Lambda }^{+}/2}\frac{A^{1}\mathrm{d}r}{%
(r-r_{\Lambda }^{+})\sqrt{(r-\alpha )(r-\beta )(r-\gamma )(r-\delta )}}
\end{equation}%
Indeed, we compute in closed-form: 
\begin{eqnarray}
&&\int_{\alpha }^{r_{\Lambda }^{+}/2}\frac{A^{1}\mathrm{d}r}{(r-r_{\Lambda
}^{+})\sqrt{(r-\alpha )(r-\beta )(r-\gamma )(r-\delta )}}  \notag \\
&=&\frac{\rho _{1}}{\sqrt{\rho _{1}}}H_{\Lambda }^{+}\times  \notag \\
&&F_{D}\left( \frac{1}{2},-1,\frac{1}{2},\frac{1}{2},1,\frac{3}{2},\frac{%
r_{\Lambda }^{+}-2\alpha }{r_{\Lambda }^{+}-2\beta },\frac{\beta -\gamma }{%
\alpha -\gamma }\frac{r_{\Lambda }^{+}-2\alpha }{r_{\Lambda }^{+}-2\beta },%
\frac{\beta -\delta }{\alpha -\delta }\frac{r_{\Lambda }^{+}-2\alpha }{%
r_{\Lambda }^{+}-2\beta },\frac{r_{\Lambda }^{+}-\beta }{r_{\Lambda
}^{+}-\alpha }\frac{r_{\Lambda }^{+}-2\alpha }{r_{\Lambda }^{+}-2\beta }%
\right) \frac{\Gamma (1/2)}{\Gamma (3/2)}  \notag \\
&&
\end{eqnarray}%
where%
\begin{equation}
\rho _{1}:=\frac{r_{\Lambda }^{+}-\beta }{r_{\Lambda }^{+}-\alpha }\frac{%
r_{\Lambda }^{+}-2\alpha }{r_{\Lambda }^{+}-2\beta }
\end{equation}

\bigskip Also the radial integral involved on the LHS in the
\textquotedblleft balance\textquotedblright\ lens equation (\ref%
{AbelIntegral}) is computed exactly in terms of the hypergeometric function
of Appell $F_{1}:$

\begin{eqnarray}
&&\frac{1}{\sqrt{\frac{1}{3}(\mathcal{Q}\Lambda +3\Xi ^{2}(1+\frac{\Lambda }{%
3}(a-\Phi )^{2})}}\int_{\alpha }^{r_{\Lambda }^{+}/2}\frac{\mathrm{d}r}{%
\sqrt{(r-\alpha )(r-\beta )(r-\gamma )(r-\delta )}}  \notag \\
&=&\frac{\rho _{1}}{\sqrt{\mathcal{E}}}\frac{1}{\sqrt{\omega (\gamma -\alpha
)(\delta -\alpha )}}\frac{\Gamma (1/2)}{\Gamma (3/2)}F_{1}\left( \frac{1}{2},%
\frac{1}{2},\frac{1}{2},\frac{3}{2},\frac{\beta -\gamma }{\alpha -\gamma }%
\frac{r_{\Lambda }^{+}-2\alpha }{r_{\Lambda }^{+}-2\beta },\frac{\beta
-\delta }{\alpha -\delta }\frac{r_{\Lambda }^{+}-2\alpha }{r_{\Lambda
}^{+}-2\beta }\right)  \notag \\
&&  \label{BalanceAbelLambda}
\end{eqnarray}%
where $\mathcal{E}:=\frac{1}{3}(\mathcal{Q}\Lambda +3\Xi ^{2}(1+\frac{%
\Lambda }{3}(a-\Phi )^{2}),$ $\omega :=\frac{r_{\Lambda }^{+}-\alpha }{%
r_{\Lambda }^{+}-\beta }$ and $\alpha ,\beta ,\gamma ,\delta $ denote the
roots of the quartic polynomial $R$ in the presence of $\Lambda $ eqn(\ref%
{quartic1}).

A complete phenomenological analysis of our exact solutions in the presence
of the cosmological constant $\Lambda $ will be a subject of a separate
publication \cite{KraniotisGV}$.$ Nevertheless, it is evident from the
closed form solutions we derived in this work that the cosmological constant 
\textit{does }contributes to the gravitational bending of light.

\section{Conclusions}

In this work the precise analytic treatment of Kerr and Kerr-de Sitter black
holes as gravitational lenses has been achieved. A full analytic
strong-field calculation of the source, image positions and the resulting
magnification factors has been performed. \ A full blend of important
functions from Mathematical Analysis such as the Weierstra\ss\ elliptic
function $\wp $ and the generalized multivariable hypergeometric functions
of Appell-Lauricella $F_{D}$ were deployed in deriving the closed-form
solution of the gravitational lens equations. From the exact solution of the
radial and angular Abelian integrals which are involved in the lens
equations we concluded the $\Lambda $ does contribute to the gravitational
bending of light. A full quantitative phenomenological analysis of
gravitational lensing by a Kerr deflector in the presence of $\Lambda $ is
beyond the scope of this work and will appear elsewhere \cite{KraniotisGV}.
We provided examples of image-source configurations that solve the
gravitational Kerr lens equations and exhibited their appearance on the
observer's image plane as they will be detected by an equatorial observer ($%
\theta _{O}=\pi /2,\phi _{O}=0$) $\ $and an observer located at $\theta
_{O}=\pi /3,\phi _{O}=0,$for various values of the Kerr parameter $a$, and
the first integrals of motion $\Phi ,\mathcal{Q}.$

The theory produced in this work based on the exact solution of the null
geodesic equations of motion in Kerr spacetime wll have an important
application to the Sgr A$^{\ast }$ galactic centre supermassive black hole 
\cite{KraniotisGV}. It may serve the important goal of probing general
relativity at the strong field regime through the phenomenon of
gravitational bending of light induced by the spacetime curvature.
It is complementary to other investigations which have the ambition to probe gravitation at the strong-field regime through the relativistic effects of periastron precession and frame-dragging \cite{KraniotisSstars}.

There is a fruitful synergy of various fields of Science: general
relativity, astronomy, cosmology and pure mathematics.

\section{Acknowledgments}

This is modified written account of the author's talk at NEB-14 Recent
Developments in Gravity that took place at the University of Ioannina. The
author would like to thank: L. Perivolaropoulos and P. Kanti for
inviting him to deliver his talk to such an exciting conference. He is also
oblidged to C. E. Vayonakis for discussions and comments on the manuscript. 
He warmly thanks his colleagues and his undergraduate students at the Physics 
department, University of Ioannina, for a stimulating academic enviroment.
In addition, he thanks G. Kakarantzas for discussions
and his friendship. Last but not least, he thanks his family for moral support 
during the early stages of this work.

\appendix

\section{Transformation properties of Lauricella's hypergeometric function $%
F_{D}\label{Veta}.$}

\bigskip

In this appendix we prove useful transformation properties of Lauricella's
hypergeometric function $F_{D}.$ We first introduce the function and its
integral representation:

\bigskip 
\fbox{Lauricella's\;$4^{th}$\;hypergeometric\;function\;of\;m-variables.}%

\begin{equation}
\fbox{$\displaystyle F_D(\alpha,{\bf \beta},\gamma,{\bf z})=
\sum_{n_1,n_2,\dots,n_m=0}^{\infty}\frac{(\alpha)_{n_1+\cdots n_m}(\beta_1)_{n_1} \cdots (\beta_m)_{n_m}}
{(\gamma)_{n_1+\cdots+n_m}(1)_{n_1}\cdots (1)_{n_m}} z_1^{n_1}\cdots z_m^{n_m}$}
\label{GLauri}
\end{equation}%

where 
\begin{eqnarray}
\mathbf{z} &=&(z_{1},\ldots ,z_{m}),  \notag \\
\mathbf{\beta } &=&(\beta _{1},\ldots ,\beta _{m}).
\end{eqnarray}

The Pochhammer symbol 
\fbox{$\displaystyle (\alpha)_m=(\alpha,m)$}
is defined by%
\begin{equation}
(\alpha )_{m}=\frac{\Gamma (\alpha +m)}{\Gamma (\alpha )}=\left\{ 
\begin{array}{ccc}
1, & 
{\rm if}
& m=0 \\ 
\alpha (\alpha +1)\cdots (\alpha +m-1) & \text{%
{\rm if}%
} & m=1,2,3%
\end{array}%
\right.
\end{equation}

The series admits the following integral representation:

\begin{equation}
\fbox{$\displaystyle F_D(\alpha,{\bf \beta},\gamma,{\bf z})=
\frac{\Gamma(\gamma)}{\Gamma(\alpha)\Gamma(\gamma-\alpha)}
\int_0^1 t^{\alpha-1}(1-t)^{\gamma-\alpha-1}(1-z_1 t)^{-\beta_1}\cdots (1-z_m t)^{-\beta_m} {\rm d}t $}
\label{OloklAnapa}
\end{equation}
which is valid for 
\fbox{$\displaystyle {\rm Re}(\alpha)>0,\;{\rm Re}(\gamma-\alpha)>0. $}%
. It 
{\em converges\;absolutely}
inside the m-dimensional cuboid:%
\begin{equation}
|z_{j}|<1,(j=1,\ldots ,m).
\end{equation}

\begin{proposition}
The following holds:%
\begin{eqnarray}
F_{D}(\alpha ,\beta ,\beta ^{\prime },\beta ^{\prime \prime },\gamma ,x,y,z)
&=&(1-z)^{-\alpha }\times  \notag \\
&&F_{D}\left( \alpha ,\beta ,\beta ^{\prime },\gamma -\beta -\beta ^{\prime
}-\beta ^{\prime \prime },\gamma ,\frac{z-x}{z-1},\frac{z-y}{z-1},\frac{z}{%
z-1}\right)  \notag \\
&&  \label{Gwniakimatesxi2FD}
\end{eqnarray}
\end{proposition}

\bigskip

\begin{proof}
Applying the tranformation in equation (\ref{IRT1}):%
\begin{equation}
u=\frac{\nu }{1-z+\nu z}=\frac{\nu }{(1-z)\left[ 1-\frac{\nu z}{z-1}\right] }
\end{equation}%
we get:%
\begin{eqnarray}
1-u &=&\frac{1-\nu }{1-\nu \frac{z}{z-1}},\quad (1-ux)^{-\beta }=\left( 
\frac{[1-\frac{\nu (z-x)}{z-1}]}{[1-\frac{\nu z}{z-1}}\right) ^{-\beta } 
\notag \\
(1-u\text{ }y)^{-\beta ^{\prime }} &=&\left( \frac{[1-\frac{\nu (z-y)}{z-1}]%
}{1-\frac{\nu z}{z-1}}\right) ^{-\beta ^{\prime }},\quad (1-uz)^{-\beta
^{\prime \prime }}=\left( \frac{1}{[1-\frac{\nu z}{z-1}]}\right) ^{-\beta
^{\prime \prime }}
\end{eqnarray}%
Thus 
\begin{eqnarray}
IRF_{D} &=&(1-z)^{-\alpha }\times  \notag \\
&&\int_{0}^{1}d\nu \text{ }\nu ^{\alpha -1}(1-\nu )^{\gamma -\alpha
-1}(1-\nu \frac{z-x}{z-1})^{-\beta }(1-\nu \frac{z-y}{z-1})^{-\beta ^{\prime
}}(1-\nu \frac{z}{z-1})^{-(\gamma -\beta -\beta ^{\prime }-\beta ^{\prime
\prime })}  \notag \\
&&
\end{eqnarray}%
and proposition follows.
\end{proof}

\bigskip

\ \ \ \ \ \ 

\begin{proposition}
\ \ \ \ \ \ \ \ The following identity holds:%
\begin{eqnarray}
&&\frac{1}{|a|}\sqrt[2]{\frac{z_{j}}{z_{m}}}\sqrt[2]{\frac{z_{m}-z_{3}}{%
z_{j}-z_{3}}}\frac{1}{\sqrt[2]{z_{m}-z_{3}}}F_{D}\left( \frac{1}{2},1,\frac{1%
}{2},-\frac{1}{2},\frac{3}{2},\frac{z_{j}(1-z_{3})}{z_{j}-z_{3}},\frac{z_{j}%
}{z_{m}}\frac{(z_{m}-z_{3})}{(z_{j}-z_{3})},\frac{z_{j}}{z_{j}-z_{3}}\right)
\notag \\
&=&\frac{1}{|a|}\sqrt[2]{\frac{z_{j}}{z_{m}}}\sqrt[2]{\frac{z_{m}-z_{3}}{%
z_{j}-z_{3}}}\frac{1}{\sqrt[2]{z_{m}-z_{3}}}\times 
\Biggl\{%
\frac{-z_{3}}{1-z_{3}}F_{D}\left( \frac{1}{2},1,\frac{1}{2},\frac{1}{2},%
\frac{3}{2},\frac{z_{j}(1-z_{3})}{z_{j}-z_{3}},\frac{z_{j}}{z_{m}}\frac{%
(z_{m}-z_{3})}{(z_{j}-z_{3})},\frac{z_{j}}{z_{j}-z_{3}}\right) +  \notag \\
&&\frac{1}{1-z_{3}}F_{1}\left( \frac{1}{2},\frac{1}{2},\frac{1}{2},\frac{3}{2%
},\frac{z_{j}}{z_{m}}\frac{(z_{m}-z_{3})}{(z_{j}-z_{3})},\frac{z_{j}}{%
z_{j}-z_{3}}\right) 
\Biggr\}
\notag \\
&&
\end{eqnarray}
\end{proposition}

\ \ \ \ \ \ \ \ \ \ \ \ \ \ \ \ \ \ 

\ \ \ \ \ \ \ \ \ \ \ \ \ \ \ \ \ \ \ \ \ \ \ \ \ \ \ \ \ \ \ \ \ \ \ \ \ \
\ \ \ \ \ \ \ \ \ \ \ \ \ \ \ \ \ \ \ \ \ \ \ \ \ \ \ \ \ \ \ \ \ \ \ \ \ \
\ \ \ \ \ \ \ \ \ \ \ \ \ \ \ \ \ \ \ \ \ \ \ \ \ \ \ \ \ \ \ \ \ \ \ \ \ \
\ \ \ \ \ \ \ \ \ \ \ \ \ \ \ \ \ \ \ \ \ \ \ \ \ \ \ \ \ \ \ \ \ \ \ \ \ \
\ \ \ \ \ \ \ \ \ \ \ \ \ \ \ \ \ \ \ \ \ \ \ \ \ \ \ \ \ \ \ \ \ \ \ \ \ \
\ \ \ \ \ \ \ \ \ \ \ \ \ \ \ \ \ \ \ \ \ \ \ \ \ \ \ \ \ \ \ \ \ \ \ \ \ \
\ \ \ \ \ \ \ \ \ \ \ \ \ \ \ \ \ \ \ \ \ \ \ \ \ \ \ \ \ \ \ \ \ \ \ \ \ \
\ \ \ \ \ \ \ \ \ \ \ \ \ \ \ \ \ \ \ \ \ \ \ \ \ \ \ \ \ \ \ \ \ \ \ \ \ \
\ \ \ \ \ \ \ \ \ \ \ \ \ \ \ \ \ \ \ \ \ \ \ \ \ \ \ \ \ \ \ \ \ \ \ \ \ \
\ \ \ \ \ \ \ \ \ \ \ \ \ \ \ \ \ \ \ \ \ \ \ \ \ \ \ \ \ \ \ \ \ \ \ \ \ \
\ \ \ \ \ \ \ \ \ \ \ \ \ \ \ \ \ \ \ \ \ \ \ \ \ \ \ \ \ \ \ \ \ \ \ \ \ \
\ \ \ \ \ \ \ \ \ \ \ \ \ \ \ \ \ \ \ \ \ \ \ \ \ \ \ \ \ \ \ \ \ \ \ \ \ \
\ \ \ \ \ \ \ \ \ \ \ \ \ \ \ \ \ \ \ \ \ \ \ \ \ \ \ \ \ \ \ \ \ \ \ \ \ \
\ \ \ \ \ \ \ \ \ \ \ \ \ \ \ \ \ \ \ \ \ \ \ \ \ \ \ \ \ \ \ \ \ \ \ \ \ \
\ \ \ \ \ \ \ \ \ \ \ \ \ \ \ \ \ \ \ \ \ \ \ \ \ \ \ \ \ \ \ \ \ \ \ \ \ \
\ \ \ \ \ \ \ \ \ \ \ \ \ \ \ \ \ \ \ \ \ \ \ \ \ \ \ \ \ \ \ \ \ \ \ \ \ \
\ \ \ \ \ \ \ \ \ \ \ \ \ \ \ \ \ \ \ \ \ \ \ \ \ \ \ \ \ \ \ \ \ \ \ \ \ \
\ \ \ \ \ \ \ \ \ \ \ \ \ \ \ \ \ \ \ \ \ \ \ \ \ \ \ \ \ \ 

\begin{proof}
We start with the integral representation of Lauricella\textquotedblright s
hypergeometric function:%
\begin{eqnarray}
F_{D}\left( \frac{1}{2},1,\frac{1}{2},-\frac{1}{2},\frac{3}{2}%
,x_{1},x_{2},x_{3}\right) &=&\int_{0}^{1}du\text{ }%
u^{-1/2}(1-ux_{1})^{-1}(1-ux_{2})^{-1/2}(1-ux_{3})^{1/2}\frac{\Gamma (3/2)}{%
\Gamma (1/2)}  \notag \\
&=&\frac{\Gamma (3/2)}{\Gamma (1/2)}\int_{0}^{1}\frac{du}{\sqrt[2]{u}}\frac{1%
}{1-ux_{1}}\frac{1}{\sqrt[2]{1-ux_{2}}}\frac{1-ux_{3}}{\sqrt[2]{1-ux_{3}}} 
\notag \\
&=&\frac{\Gamma (3/2)}{\Gamma (1/2)}%
\Biggl[%
\int_{0}^{1}\frac{du}{\sqrt[2]{u}(1-ux_{1})}\frac{1}{\sqrt[2]{1-ux_{2}}}%
\frac{1}{\sqrt[2]{1-ux_{3}}}-  \notag \\
&&\int_{0}^{1}\frac{du\text{ }u\text{ }x_{3}}{\sqrt[2]{u}(1-ux_{1})}\frac{1}{%
\sqrt[2]{1-ux_{2}}}\frac{1}{\sqrt[2]{1-ux_{3}}}%
\Biggr]%
\end{eqnarray}%
with $x_{1}:=\frac{z_{j}(1-z_{3})}{z_{j}-z_{3}},x_{2}:=\frac{z_{j}}{z_{m}}%
\frac{(z_{m}-z_{3})}{(z_{j}-z_{3})},x_{3}:=\frac{z_{j}}{z_{j}-z_{3}}.$
\end{proof}

\section{\protect\bigskip Time-delay assuming vanishing $\Lambda $.}

For the time-delay, in the case of vansihing cosmological constant, we
derive the equation:%
\begin{equation}
ct=\int^{r}\frac{r^{2}(r^{2}+a^{2})}{\pm \Delta \sqrt{R}}\mathrm{d}r+\int^{r}%
\frac{2GMr}{\pm c^{2}\Delta \sqrt{R}}(a^{2}-\Phi a)\mathrm{d}r+\int^{\theta }%
\frac{a^{2}\cos ^{2}\theta \mathrm{d}\theta }{\pm \sqrt{\Theta }}
\label{Timedelay}
\end{equation}%
In calculating the last angular term in (\ref{Timedelay}) and using the
variable $z=\cos ^{2}\theta ,$ one of the integrals we need to calculate is:%
\begin{equation}
\frac{1}{2}\frac{a^{2}}{|a|}\int_{0}^{z_{j}}\frac{\mathrm{d}z\text{ }z}{%
\sqrt{z(z_{m}-z)(z-z_{3})}}
\end{equation}%
Indeed, its calculation in closed analytic form gave us the result:%
\begin{eqnarray}
&&\frac{1}{2}\frac{a^{2}}{|a|}\int_{0}^{z_{j}}\frac{\mathrm{d}z\text{ }z}{%
\sqrt{z(z_{m}-z)(z-z_{3})}}  \notag \\
&=&\frac{1}{2}\frac{a^{2}}{|a|}\frac{z_{j}^{2}(z_{m}-z_{j})}{z_{m}\sqrt{%
(z_{j}-z_{m})(z_{3}-z_{j})z_{j}}}F_{1}\left( 1,\frac{3}{2},\frac{1}{2},\frac{%
5}{2},\frac{z_{j}}{z_{m}},\frac{z_{j}}{z_{m}}\frac{z_{m}-z_{3}}{z_{j}-z_{3}}%
\right) \frac{\Gamma (1)\Gamma (3/2)}{\Gamma (5/2)}  \notag \\
&&
\end{eqnarray}%
In total we derive for the angular integrals in (\ref{Timedelay}):%
\begin{eqnarray}
\int^{\theta }\frac{a^{2}\cos ^{2}\theta \mathrm{d}\theta }{\pm \sqrt{\Theta 
}} &\equiv &A^{\mathrm{time-delay}}=  \notag \\
&&\frac{1}{2}\frac{a^{2}}{|a|}\frac{(z_{m}-z_{S})z_{m}}{\sqrt{%
z_{m}(z_{m}-z_{S})(z_{m}-z_{3})}}F_{1}\left( \frac{1}{2},-\frac{1}{2},\frac{1%
}{2},\frac{3}{2},\frac{z_{m}-z_{S}}{z_{m}},\frac{z_{m}-z_{S}}{z_{m}-z_{3}}%
\right) \frac{\Gamma (1/2)}{\Gamma (3/2)}+  \notag \\
&&[1-\mathrm{sign}(\theta _{S}\circ \theta _{mS})]\frac{1}{2}\frac{a^{2}}{|a|%
}\frac{z_{S}^{2}(z_{m}-z_{S})}{z_{m}\sqrt{(z_{S}-z_{m})(z_{3}-z_{S})z_{S}}} 
\notag \\
&&\times F_{1}\left( 1,\frac{3}{2},\frac{1}{2},\frac{5}{2},\frac{z_{S}}{z_{m}%
},\frac{z_{S}}{z_{m}}\frac{z_{m}-z_{3}}{z_{S}-z_{3}}\right) \frac{\Gamma
(1)\Gamma (3/2)}{\Gamma (5/2)}  \notag \\
&&+\frac{1}{2}\frac{a^{2}}{|a|}\frac{(z_{m}-z_{O})z_{m}}{\sqrt{%
z_{m}(z_{m}-z_{O})(z_{m}-z_{3})}}F_{1}\left( \frac{1}{2},-\frac{1}{2},\frac{1%
}{2},\frac{3}{2},\frac{z_{m}-z_{O}}{z_{m}},\frac{z_{m}-z_{O}}{z_{m}-z_{3}}%
\right) \frac{\Gamma (1/2)}{\Gamma (3/2)}+  \notag \\
&&[1-\mathrm{sign}(\theta _{O}\circ \theta _{mO})]\frac{1}{2}\frac{a^{2}}{|a|%
}\frac{z_{O}^{2}(z_{m}-z_{O})}{z_{m}\sqrt{(z_{O}-z_{m})(z_{3}-z_{O})z_{O}}} 
\notag \\
&&\times F_{1}\left( 1,\frac{3}{2},\frac{1}{2},\frac{5}{2},\frac{z_{O}}{z_{m}%
},\frac{z_{O}}{z_{m}}\frac{z_{m}-z_{3}}{z_{O}-z_{3}}\right) \frac{\Gamma
(1)\Gamma (3/2)}{\Gamma (5/2)}  \notag \\
&&+2(m-1)\frac{1}{2}\frac{a^{2}}{|a|}\frac{z_{m}^{2}}{\sqrt{%
z_{m}^{2}(z_{m}-z_{3})}}\frac{\Gamma (3/2)\Gamma (1/2)}{\Gamma (2)}F\left( 
\frac{1}{2},\frac{1}{2},2,\frac{z_{m}}{z_{m}-z_{3}}\right)
\end{eqnarray}

We now turn our attention to the calculation of the radial contribution to
time-delay in equation (\ref{Timedelay}) . Indeed, the first term can be
written:%
\begin{eqnarray}
&&\int_{\alpha }^{r_{S}}\frac{r^{2}(r^{2}+a^{2})}{\Delta \sqrt{R}}\mathrm{d}r
\notag \\
&=&\int_{\alpha }^{r_{S}}\frac{r^{2}\mathrm{d}r}{\sqrt{R}}+\int_{\alpha
}^{r_{S}}\frac{2GMr}{c^{2}\sqrt{R}}\mathrm{d}r-\int_{\alpha }^{r_{S}}\frac{%
2a^{2}GMr\mathrm{d}r}{c^{2}\Delta \sqrt{R}}+\frac{4G^{2}M^{2}}{c^{4}}%
\int_{\alpha }^{r_{S}}\frac{\left( 1-\frac{a^{2}-2GMr}{\Delta }\right) }{%
\sqrt{R}}\mathrm{d}r  \notag \\
&&
\end{eqnarray}%
In total for this radial term the exact integration yields the result:%
\begin{eqnarray*}
\int_{\alpha }^{r_{S}}\frac{r^{2}(r^{2}+a^{2})}{\Delta \sqrt{R}}\mathrm{d}r
&=&\frac{\alpha ^{2}\Omega ^{\prime \prime }\mathbf{z}_{S}}{\sqrt{\frac{%
r_{S}-\alpha }{r_{S}-\beta }}} \\
&&\times F_{D}\left( \frac{1}{2},-2,2,\frac{1}{2},\frac{1}{2},\frac{3}{2},%
\frac{\beta }{\alpha }\frac{r_{S}-\alpha }{r_{S}-\beta },\frac{r_{S}-\alpha 
}{r_{S}-\beta },\frac{r_{S}-\alpha }{r_{S}-\beta }\frac{\beta -\gamma }{%
\alpha -\gamma },\frac{\delta -\beta }{\delta -\alpha }\frac{r_{S}-\alpha }{%
r_{S}-\beta }\right) \frac{\Gamma (1/2)}{\Gamma (3/2)} \\
&&+\frac{\alpha \Omega ^{\prime \prime }\mathbf{z}_{S}}{\sqrt{\frac{%
r_{S}-\alpha }{r_{S}-\beta }}}\frac{2GM}{c^{2}}F_{D}\left( \frac{1}{2},-1,1,%
\frac{1}{2},\frac{1}{2},\frac{3}{2},\frac{\beta }{\alpha }\frac{r_{S}-\alpha 
}{r_{S}-\beta },\frac{r_{S}-\alpha }{r_{S}-\beta },\frac{r_{S}-\alpha }{%
r_{S}-\beta }\frac{\beta -\gamma }{\alpha -\gamma },\frac{1}{\omega }\frac{%
r_{S}-\alpha }{r_{S}-\beta }\right) 2 \\
&&-\frac{\mathbf{z}_{S}\Omega ^{\prime \prime }A_{+}^{td}}{(r_{+}-\alpha )%
\sqrt{\frac{r_{S}-\alpha }{r_{S}-\beta }}} \\
&&\times 
\Biggl\{%
F_{D}\left( \frac{1}{2},1,\frac{1}{2},\frac{1}{2},\frac{3}{2},\frac{\beta
-r_{+}}{\alpha -r_{+}}\frac{r_{S}-\alpha }{r_{S}-\beta },\frac{r_{S}-\alpha 
}{r_{S}-\beta }\frac{\beta -\gamma }{\alpha -\gamma },\frac{\delta -\beta }{%
\delta -\alpha }\frac{r_{S}-\alpha }{r_{S}-\beta }\right) \frac{\Gamma (1/2)%
}{\Gamma (3/2)} \\
&&-\frac{r_{S}-\alpha }{r_{S}-\beta }F_{D}\left( \frac{3}{2},1,\frac{1}{2},%
\frac{1}{2},\frac{5}{2},\frac{\beta -r_{+}}{\alpha -r_{+}}\frac{r_{S}-\alpha 
}{r_{S}-\beta },\frac{r_{S}-\alpha }{r_{S}-\beta }\frac{\beta -\gamma }{%
\alpha -\gamma },\frac{\delta -\beta }{\delta -\alpha }\frac{r_{S}-\alpha }{%
r_{S}-\beta }\right) \frac{\Gamma (3/2)}{\Gamma (5/2)}%
\Biggr\}
\\
&&-\frac{\mathbf{z}_{S}\Omega ^{\prime \prime }A_{-}^{td}}{(r_{-}-\alpha )%
\sqrt{\frac{r_{S}-\alpha }{r_{S}-\beta }}} \\
&&\times 
\Biggl\{%
F_{D}\left( \frac{1}{2},1,\frac{1}{2},\frac{1}{2},\frac{3}{2},\frac{\beta
-r_{-}}{\alpha -r_{-}}\frac{r_{S}-\alpha }{r_{S}-\beta },\frac{r_{S}-\alpha 
}{r_{S}-\beta }\frac{\beta -\gamma }{\alpha -\gamma },\frac{\delta -\beta }{%
\delta -\alpha }\frac{r_{S}-\alpha }{r_{S}-\beta }\right) \frac{\Gamma (1/2)%
}{\Gamma (3/2)} \\
&&-\frac{r_{S}-\alpha }{r_{S}-\beta }F_{D}\left( \frac{3}{2},1,\frac{1}{2},%
\frac{1}{2},\frac{5}{2},\frac{\beta -r_{-}}{\alpha -r_{-}}\frac{r_{S}-\alpha 
}{r_{S}-\beta },\frac{r_{S}-\alpha }{r_{S}-\beta }\frac{\beta -\gamma }{%
\alpha -\gamma },\frac{\delta -\beta }{\delta -\alpha }\frac{r_{S}-\alpha }{%
r_{S}-\beta }\right) \frac{\Gamma (3/2)}{\Gamma (5/2)}%
\Biggr\}
\\
&&+\frac{4G^{2}M^{2}}{c^{4}}\frac{\Omega ^{\prime \prime }\mathbf{z}_{S}}{%
\sqrt{\frac{r_{S}-\alpha }{r_{S}-\beta }}}F_{1}\left( \frac{1}{2},\frac{1}{2}%
,\frac{1}{2},\frac{3}{2},\frac{r_{S}-\alpha }{r_{S}-\beta }\frac{\beta
-\gamma }{\alpha -\gamma },\frac{\delta -\beta }{\delta -\alpha }\frac{%
r_{S}-\alpha }{r_{S}-\beta }\right) \frac{\Gamma (1/2)}{\Gamma (3/2)} \\
&&-\frac{4G^{2}M^{2}}{c^{4}}A_{+1}^{td}\frac{\Omega ^{\prime \prime }(-%
\mathbf{z}_{S})}{(r_{+}-\alpha )\sqrt{\frac{r_{S}-\alpha }{r_{S}-\beta }}} \\
&&\times 
\Biggl\{%
F_{D}\left( \frac{1}{2},1,\frac{1}{2},\frac{1}{2},\frac{3}{2},\frac{\beta
-r_{+}}{\alpha -r_{+}}\frac{r_{S}-\alpha }{r_{S}-\beta },\frac{r_{S}-\alpha 
}{r_{S}-\beta }\frac{\beta -\gamma }{\alpha -\gamma },\frac{\delta -\beta }{%
\delta -\alpha }\frac{r_{S}-\alpha }{r_{S}-\beta }\right) \frac{\Gamma (1/2)%
}{\Gamma (3/2)} \\
&&-\frac{r_{S}-\alpha }{r_{S}-\beta }F_{D}\left( \frac{3}{2},1,\frac{1}{2},%
\frac{1}{2},\frac{5}{2},\frac{\beta -r_{+}}{\alpha -r_{+}}\frac{r_{S}-\alpha 
}{r_{S}-\beta },\frac{r_{S}-\alpha }{r_{S}-\beta }\frac{\beta -\gamma }{%
\alpha -\gamma },\frac{\delta -\beta }{\delta -\alpha }\frac{r_{S}-\alpha }{%
r_{S}-\beta }\right) \frac{\Gamma (3/2)}{\Gamma (5/2)}%
\Biggr\}
\\
&&-\frac{4G^{2}M^{2}}{c^{4}}A_{-1}^{td}\frac{\Omega ^{\prime \prime }(-%
\mathbf{z}_{S})}{(r_{-}-\alpha )\sqrt{\frac{r_{S}-\alpha }{r_{S}-\beta }}} \\
&&\times 
\Biggl\{%
F_{D}\left( \frac{1}{2},1,\frac{1}{2},\frac{1}{2},\frac{3}{2},\frac{\beta
-r_{-}}{\alpha -r_{-}}\frac{r_{S}-\alpha }{r_{S}-\beta },\frac{r_{S}-\alpha 
}{r_{S}-\beta }\frac{\beta -\gamma }{\alpha -\gamma },\frac{\delta -\beta }{%
\delta -\alpha }\frac{r_{S}-\alpha }{r_{S}-\beta }\right) \frac{\Gamma (1/2)%
}{\Gamma (3/2)} \\
&&-\frac{r_{S}-\alpha }{r_{S}-\beta }F_{D}\left( \frac{3}{2},1,\frac{1}{2},%
\frac{1}{2},\frac{5}{2},\frac{\beta -r_{-}}{\alpha -r_{-}}\frac{r_{S}-\alpha 
}{r_{S}-\beta },\frac{r_{S}-\alpha }{r_{S}-\beta }\frac{\beta -\gamma }{%
\alpha -\gamma },\frac{\delta -\beta }{\delta -\alpha }\frac{r_{S}-\alpha }{%
r_{S}-\beta }\right) \frac{\Gamma (3/2)}{\Gamma (5/2)}%
\Biggr\}
\\
&&
\end{eqnarray*}

\end{document}